\documentclass[conference,letterpaper]{IEEEtran}

\usepackage[utf8]{inputenc} 
\usepackage[T1]{fontenc}
\usepackage{url}
\usepackage{ifthen}
\usepackage{cite}
\usepackage[cmex10]{amsmath} 

\usepackage{epsfig}
\usepackage{framed}

\usepackage{xcolor}

\usepackage[utf8]{inputenc}
\usepackage{amssymb}
\usepackage{amsthm}
\usepackage{amsmath}
\usepackage{amsfonts}
\usepackage{graphicx}
\usepackage{comment}
\usepackage{float}
\usepackage{subfigure}
\usepackage{algpseudocode}
\usepackage[ruled]{algorithm2e}
\usepackage{tabularx}

\usepackage{enumitem}
\usepackage{mathrsfs}
\usepackage{graphicx}
\usepackage{multirow}
\usepackage{color,soul}
\usepackage{color}
\usepackage{float}
\usepackage{enumitem}
\usepackage{adjustbox}
\usepackage{algorithmicx}
\usepackage{hyperref}
\usepackage{algcompatible}

\allowdisplaybreaks
\DeclareMathAlphabet{\mathpzc}{OT1}{pzc}{m}{it}
\algrenewcommand\alglinenumber[1]{\scriptsize #1:}
\definecolor{mygray}{gray}{0.6}

\newcommand{\pr}[1]{\mathrm{Pr}#1} 

\newcommand{\mc}[1]{\mathcal{#1}}

\newcommand{\ms}[1]{\mathsf{#1}}

\newcommand{\remove}[1]{}

\newcommand{\ceil}[1]{\lceil #1 \rceil}
\newcommand{\floor}[1]{\lfloor #1 \rfloor}

\providecommand{\eqref}[1]{(\ref{#1})}

\newtheorem{definition}{Definition}
\newtheorem{lemma}{Lemma}
\newtheorem{theorem}{Theorem}

\newtheorem{construction}{Construction}

\newcommand{\x}{\mathbf{x}}
\newcommand{\X}{\mathbf{X}}
\newcommand{\y}{\mathbf{y}}
\newcommand{\Y}{\mathbf{Y}}
\newcommand{\z}{\mathbf{z}}
\newcommand{\Z}{\mathbf{Z}}

\newcommand{\mg}{\color{magenta}}

\title{
 Robust and Reusable 
 Fuzzy Extractors for Low-entropy Rate Randomness Sources 
}

\author{\IEEEauthorblockN{Somnath Panja}
\IEEEauthorblockA{University of Calgary, Canada}
\and
\IEEEauthorblockN{Shaoquan Jiang}
\IEEEauthorblockA{University of Windsor, Canada}
\and
\IEEEauthorblockN{Reihaneh Safavi-Naini}
\IEEEauthorblockA{University of Calgary, Canada}
}

\begin{document}

\maketitle

\begin{abstract}
Fuzzy  extractors (FE) 
are cryptographic primitives that extract reliable  cryptographic key from  
noisy real world random sources such as biometric sources.
%
%
The FE generation algorithm takes a source sample, extracts a key and generates some {\em helper} data that will be used by the reproduction algorithm to recover the key. Reusability of FE 
guarantees
that security holds when 
FE is used 
multiple times  with the same source, 
and robustness of FE 
requires tampering with  the helper data 
be detectable.
\remove{
Fuzzy extractors have been studied in
information theoretic 
and computational setting.
}

%
In this paper, we consider information theoretic FEs, 
define a strong notion of reusability, 
and propose 
 {\em strongly robust and reusable FEs} (srrFE) that provides  the strongest combined  notion of reusability and robustness for FEs. 
We give two constructions, one for reusable FEs and one 
 for srrFE 
with information theoretic (IT) security for structured sources. 
\remove{Constructions 1 and 2 provide information theoretic (IT) security, in our proposed model.}
%
%
 The constructions are for structured sources and use {\em sample-then-lock} approach.  We discuss each construction and show their unique properties in relation to existing work. 
 
 Construction 2 is the  first 
 robust and reusable FE  with IT-security without assuming random oracle. The robustness is achieved by using an IT-secure MAC  with security against key-shift attack, which can be of independent interest.

\end{abstract}



\section{Introduction}\label{sec:intro}
%
%
Secret key establishment (SKE)  is a fundamental problem in cryptography.  A general model of SKE in information theoretic setting is  when Alice, Bob and Eve have  random variables   $X$, $Y$,
and $Z$, respectively, with a joint probability distribution $P_{XYZ}$, where $Z$ summarizes Eve's information about $X$ and $Y$. Maurer~\cite{Maurer1993}  
proved that in this setting the secret key entropy is upper-bounded by $\min [I(X;Y ), I ( X;Y |Z)]$ (ignoring small constants) and  so  key establishment is possible only if $X$ and $Y$ are correlated. 
There is a large body of research on general setting 
when Alice and Bob can communicate over a public authenticated channel.
 A special case of this general  setting is  called {\em Fuzzy
Extractors} (FEs)  \cite{DodisORS08} setting,  in which 
Alice and Bob's variables are   samples of a ``noisy"  source, and $d(x,y)\leq t$, where $d$ is a distance metric. 
Eve does not have any initial information about $X$ and $Y$, and so $Z$ is null.

\remove{
{\em Fuzzy
extractors} (FEs)  \cite{DodisORS08}, 
establish a shared key between two parties, Alice and Bob, who have access to two 
``close" samples of a random source, where closeness is with respect to some distance metric,  and can communicate over a public channel.
}

The  randomness source in  FE  (generating $X$and $Y$) naturally occurs in sampling biometrics and behavioral data  \cite{Daugman2004,karakaya2019using,islam2021scalable}, and so FE has been widely studied for biometric sources and in particular for  biometric based authentication.
  In this paper we consider IT-secure FEs.

An FE   has 
a pair of randomized  algorithms, $(Gen, Rep)$ that works as follows:  
$Gen$ takes Alice's sample  $w$ and generates a pair $(R,P)$ where $R$ is the secret key and  $P$  is  
 the {\em helper data} that will be sent to Bob, allowing  him  to use  
it in the $Rep$ 
algorithm  together with his 
sample  $w'$ that is ``close" to $w$, 
 and reproduce the same secret key $R$.

{\em Correctness 
and security} are two basic properties of FEs. {\em Correctness} of FE 
requires that  for a pair $(R,P)$ that is the output of the $Gen$ algorithm, the $Rep$ algorithm  recover $R$ from $w'$ and $P$ with 
a probability at least $1-\epsilon$, where probability is over the 
randomness of the $Gen$ algorithm.
 {\em Security} of FE requires that
  the key $R$  be 
  indistinguishable from a uniformly random string of the same length, given the adversary's view of the communication (that is $P$).


\noindent
{\bf Reusability and Robustness.} To use FE in practice, two additional properties are required.

{\em Reusability} of FE  \cite{boyen2004reusable} considers security of FE output (randomness of $R$) 
when  the 
 source is used multiple times, and {\em robustness} requires detection of tampering of the helper data.

\vspace{0.5mm}
Reusability   is important when the same source is used multiple times. 
For example, biometric  scans of a user
are used for enrolments to multiple organizations. 
%
 Reusability is defined by a  game between a {\em challenger} and  the {\em adversary Eve}, where Eve can 
ask (query) the challenger  to provide the output of the $Gen$ algorithm on   
source samples $w^1, w^2, \cdots w^\eta$, 
and receive 
the output of the $Gen$ algorithm, either the full $(R,P)$ or $P$ only, to them.  Resuability requires that 
the value  $R$ for an output of  $Gen$  algorithm that is not directly given to Eve, remains indistinguishable from random.
Robustness is defined by allowing the same type of queries by Eve, however 
the goal of the adversary is   to tamper with the helper data (e.g. 
so that Bob
recovers a different key).

Variations of reuseability and robustness definitions capture different 
powers of Eve.
\remove{
These variations 
can be grouped  using 
Eve's allowed type of queries (queried samples) 
(R1), and the information that  is 
returned to Eve (R2),  each further refined as follows.
}
For a secret sample $w$,  two types of queries have been considered. In   
{\bf (R1.1)},  Eve  can only  choose  a shift value $d^i$, and receive the output of $Gen$ on $w^i = w+d^i$ 
 \cite{boyen2004reusable}. 
 (In Boyen et al's definition of reusability,  Eve can choose an arbitrary  distortion function that is applied on $w$.   Known constructions 
 are only for
shift function.) 
In
{\bf (R1.2)} 
Eve is allowed to choose  samples  that are arbitrarily correlated with $w$ 
\cite{Canetti2020}. 
In both cases $w^i$ must be in ``close" distance of $w$ for $Rep$ to succeed.

Reusability notions also vary by the type of information that Eve receives in response to their queries. In
{\bf (R2.1)}  Eve only receives  
$P^i$ (the public output of the $Gen$ on $w^i$), and in
{\bf (R2.2)}  Eve is given the full  output of $Gen$ on $w^i$, that is $(P^i, R^i)${\footnote{In~\cite{boyen2004reusable}, {\it insider security model} allows an adversary to obtain $\tilde{R}^i \gets Rep(\tilde{P}^i,w+d^i)$, where $\tilde{P}^i$ and $d^i$ are chosen by the adversary.\label{ftnotereussaabliltyboyen}}}.

\vspace{0.5mm}
{\em Robustness}  is also defined using similar variations in Eve's queries and   received responses. 
{\it Pre-application robustness} and {\it post-application robustness}, respectively, correspond to the case that
Eve has access to $P$ only (R2.1)  or  sees $(P,R)$ (R2.2),   before modifying $P$ and 
generating $P'$ ($\neq P$)  that must be accepted by the $Rep$ algorithm ~\cite{dodis2006robust}.



{\em Robust and reusable FE (rrFE)}  require 
that robustness  hold over multiple applications of FE.
In $\eta$ robust and resuable FE,  Eve can query the $Gen$ algorithm on $\eta$  samples, 
and succeeds if it can modify  one of the $P_i$'s  without being detected by the $Rep$ algorithm.
Different flavours of reusuability and robustness that are defined above can be straightforwardly extended to the case of $\eta$ samples.

Wen et al.\cite{WenLAsiaCr18,WenLH18,Wen2019} considered  rrFE model when Eve can  only choose samples that are  shifts  of $w$ (R1.1).

\noindent
{\bf Fuzzy Extractor Constructions}. 
 Fuzzy extractors have been primarily studied for min-entropy sources where the min-entropy of the source is $\beta$.

  An established approach to construct an FE 
is to use a {\em secure sketch} algorithm 
that generates a {\em helper} string that does not loose ``too much" entropy, but allows $w$ to be recovered by Bob.
Alice and Bob both obtain the same key by applying an extractor (specified by a random seed) to obtain the same shared key.
This approach is referred to as  {\em sketch-and-extract} paradigm.  Boyen~\cite{boyen2004reusable} gave  two  general constructions of reusable FEs both using this approach and employing error correcting codes.  The first construction uses a randomly selected hash function from a  class of pairwise independent hash functions to extract randomness from the sample $w$. The second construction uses a  cryptographic hash to extract randomness from the sample $w$, and is modelled as a random oracle in the proof.

\remove{
There are also specific approaches that can be applied in special cases, for example \cite[Construction 2]{Canetti2020} that is used for large alphabets and is not reusable.
}

Sketch-and-extract approach however does not work for  {\em low-entropy sources} that are too ``noisy". That is  the ratio of the min-entropy of a sample to its size is small, and the distance  
between $w$ and $w'$ is large. In this case the $Rep$ algorithm will fail  because the helper data will loose all the sample's entropy.   Canetti et al.~\cite[Proposition 1]{Canetti2020} showed that (for binary sources) to produce non-zero key length, the min-entropy of a  sample must be at least $t\log_2\frac{n}{t}$
where $n$ is the sample size and 
$t$ is the upperbound on $d(w,w')$,
 the error in the sample that must be corrected.
 %
 %
 It was shown that an example of low entropy sources  can be  
 obtained by sampling iris data \cite{SimhadriSF19}.

Canetti et al.~\cite{Canetti2020}
constructed a {\em computationally secure reusable FE} using 
a new  approach, called {\em Sample-then-Lock}, that works for 
a class of {\em structured sources}  that is called {\em sources with high-entropy samples}.  The construction uses an idealized cryptographic primitive that is called {\em digital locker} and in the construction in \cite{Canetti2020}, is instantiated by a hash function.
The proof of security is in the random oracle model and abstract the hash function as a {\em random oracle}.
%
The approach works for  structured sources in which random subsamples of a sample  have sufficiently high min-entropy.

\subsection{Our work}
  We define IT-secure  {\em strongly robust and reusable FE (srrFE)}, and construct first, a  {\em  reusable FE (Construction~\ref{const1})  }
 and then  extend it to an {\em IT-secure srrFE (Construction~\ref{const2:cca})}. 
  We use the term `strongly'   
  to  emphasize 
  that Eve's queries can be arbitrarily correlated with $w$ (R1.2),  and reusability is in the sense of R2.2. 
 The srrFE is the first and the only known construction of an IT secure srrFE in {\em standard model} and without assuming 
 random oracles \cite{Canetti2020}.
 
 %
%

Both constructions are for {\em structured sources that are called    $(\alpha, m,N)$ conditional entropy sources} (Definition \ref{defn:condsource}). In these sources 
  random sub-samples of length $m$ of a sample $w$, conditioned on substrings of length $N$ of $w$, have at least $\alpha$ bits of  entropy. 
For $N=0$ we obtain the structured source that was used in \cite{Canetti2020}. 
%
{\em 
Robustness} is  achieved in {\em CRS  (Common Randomness String) model}. In this model, parties have access to 
 a public uniformly   random string that is independent of the source sample. 
\remove{of the source and its distribution. This model is   also used  in other known constructions of  robust and reusable FE 

by using a  key-shift secure MAC to ensure that Eve cannot exploit linearity of the helper data. 
}
CRS model had also  been used in 
 all known rrFEs \cite{WenLAsiaCr18,Wen2019,Liu2021}  that are 
   computationally  secure, and so one expects the model when the adversary is more powerful (information theoretic). 

%
%
 {\em Overview of the construction.}  
We  use a {\em strong average case  extractor} to extract uniform randomness from subsamples of the source (that are guaranteed to have sufficient min-entropy) and 
use  them as the pad in a One-Time-Pad (OTP) to encrypt, 
a randomly chosen key $R$. Alice uses $\ell$ subsamples and encrypts $R$   $\ell$ times.   Bob uses the corresponding  subsamples of $w'$ (position bits of subsamples are sent to Bob) to recover $R$. Successful $Rep$ requires at least one pair of subsamples in $w$  and $w'$ to be identical.
The parameter $\ell$ is selected to guarantee that this happens with high probability.

 Our security proof of FE effectively shows  encryption security of the 
 composition of multiple OTP of the   message ($R$); that is using the extracted randomness 
 from  the $\ell$  subsamples that are correlated,  
and using them in $\ell$ times OTP of the same message $R$ does not leak any information about $R$ to Eve. 

{\bf Construction 2} is based on Construction 1, and is  the first IT-secure {\em strongly robust and reusable fuzzy extractor (srrFE)} in standard model.
 The construction uses the CRS to allow Alice and Bob to share the same bit indexes for subsamples.
 Robustness is achieved by using a MAC (Message Authentication Code)  tag that is calculated on the set of $\ell$ ciphertexts, and is constructed using an  IT-secure one-time MAC algorithm.  
 The MAC must satisfy {\em key-shift security} 
 and protect  
 against an  adversary that can modify both the message and the key of the MAC. 
We use 
the MAC in 
\cite{rrFE2023}  and for completeness prove its security  in 
Appendix \ref{pf:lemmamac}.
Theorem \ref{THM:FUZZYEXTOTCONST2} proves that the composition of our rFE and the MAC is an srrFE.
Comparison of our constructions with existing works is in Table~\ref{table:comparison}.

{\bf Communication and computation cost.} 
In Construction~\ref{const1}, the helper data size 
is 
 %
{\footnotesize $\sim ( \ell \cdot (\alpha + 2 - 2\log(\frac{2\ell}{\sigma}) + m\ceil{\log(n)}))$} 
bits for  reusable FE, where $\ell$ is the number of subsamples 
used in the construction, 
$\sigma$ is the advantage of the adversary in the reusability game (Definition~\ref{def:rfuzzyext}), and $m, \alpha$ and $N$, are 
parameters of $(\alpha, m,N)$ conditional entropy sources  (Definition~\ref{defn:condsource}). 

%
In Construction~\ref{const2:cca},      the CRS  and the helper data sizes are {\footnotesize $\sim (\ell m\ceil{\log(n)} + m)$} and {\footnotesize$\sim [(\ell \cdot (\alpha + 2 - 2\log(\frac{2\ell}{\sigma})) + \mbox{(length of MAC tag})]$}   respectively.

Both constructions are computationally efficient. Construction~\ref{const1} requires 
$\sim \ell$
computations of a universal hash function that is used as a strong average case  extractor,   and  Construction~\ref{const2:cca} has the same computation with an additional MAC tag computation. 

 The only other reusable FEs with IT security are 
two general constructions due to Boyen \cite{boyen2004reusable} and 
the relationship between $m, n$ and $t$, the
min-entropy, sample size and distance bound, needs instantiation of the building blocks. 
There is no other IT-secure srrFE.

\vspace{-.3em}
\subsection{Related work}
\label{ap:related}
\vspace{-.5em}
Fuzzy extractors have been  extensively studied  
  in  information theoretic \cite{dodis2006robust,boyen2004reusable,boyen2005secure,fuller2020fuzzy} and computational settings \cite{Canetti2020,WenLH18,AponCEK17,WenLAsiaCr18,Wen2019,feng2021computational,Apon2022}. 
 Table \ref{table:comparison} summarizes properties  of  IT-secure   schemes that are directly related to our work. 

\vspace{-.5em}
\begin{table}[!ht]
\footnotesize
 \begin{center}
\begin{tabular}{  p{1.6cm}  c c c c c c }
\hline
FE Scheme & Distr & ReUse 
& QuCor
& Rbst
& 
Sec
\\
\hline
\cite{DodisORS08}  & minEnt  & $-$ & $-$ & $-$ & Std 
\\
\cite{boyen2004reusable}  & minEnt  & R2.1 & Shift (R1.1)& $-$ & Std 
\\
\cite{boyen2004reusable}  & minEnt & R2.2\footref{ftnotereussaabliltyboyen} & Shift (R1.1) & $-$ & 
RO\\  
\cite{boyen2005secure}  & minEnt  & $-$ & $-$ & R2.1 & 
RO\\
\cite{dodis2006robust,CramerDFPW08}  & minEnt  & $-$ & $-$ & R2.2 & Std 
\\
\cite[Constr. 3]{Canetti2020}  & lAlph  & $-$ & $-$ & $-$ & Std \\
Construction~\ref{const1}  & Strct  & R2.2 & Arbitrary (R1.2) & $-$ & Std 
\\
 Construction~\ref{const2:cca}  & Strct  & R2.2 & Arbitrary (R1.2) & R2.2 & Std  
 \\
 \hline
\end{tabular}
\end{center}
\caption{\footnotesize \textrm{
Constructions of information theoretic fuzzy extractor. 
Columns are defined as: 
Distr:  Required source 
distribution;
of the source in the  scheme and its type;
ReUse: Reusability type;
 QuCor: Query correlation type;
 Rbst: Robustness type;
 Sec: 
 Security 
 model;
 minEnt: Bound on  source min-entropy;
 lAlph:  
 Large alphabet source;
 Strct: Structured source;
 Std: 
 Standard model; 
 RO: 
 Random Oracle model;
 ``$-$'': 
 not achieved; Constr. 3 denotes Construction 3 in~\cite{Canetti2020}. 
%
}
}
\vspace{-2.3em}
\label{table:comparison}
\end{table}

  Canetti et al. \cite[Proposition 1]{Canetti2020}  gave the lower-bound $\sim t\log(\frac{n}{t})$  on the required sample entropy of  FE when the distance between two samples is bounded by $t$, when the $Rep$ algorithm receives the correct helper data.
\remove{for  key establishment to be successful, assuming 
the $Rep$  algorithm has access to correct helper data, or equivalently there is a public authenticated channel between Alice who runs $Gen$ algorithm, and Bob who uses $Rep$ algorithm to recover the key.
}
 To achieve robustness in FE however, the min-entropy  of a sample of size $n$, must be at least $n/2$ ~\cite[Section  5]{dodis2009non}.
 \remove{however it is proved that to achieve robustness, which allows key establishment to succeed by detecting tampering with the helper data, 
the min-entropy  of a sample of size $n$, must be at least $n/2$.
}
The entropy requirement of robust FE is inline with known results for key establishment in  $P_{XYZ}$ setting.  
In \cite{MaurerW03b} it was proved that when Alice and Bob share an $n$-bit string, key extraction in presence of an active adversary is possible if (R\'enyi) entropy of the shared string is lower bounded by $2n/3$ bit. 
This result was later improved to $n/2$ \cite{dodis2006robust}, matching robust FE result in \cite{dodis2009non}.  The bounds are for general sources and do not apply to our 
constructions 
that are for structured sources.

\vspace{-.5em}
\section{ Background  and Definitions
} 
\label{sec:pre}
\vspace{-0.5em}
\noindent
{\em Notations.}
 Upper-case letters (e.g., $X$) refer to random variables (RVs), and lower-case letters (e.g., $x$) 
 denote their realizations. $\mathrm{P}_X$ denotes the probability distribution 
  of an RV 
	   $X$. 
    All logarithms are base $2$ unless specified. 
The  \emph{min-entropy} of an RV $X \in \mathcal{X}$ with distribution $\mathrm{P}_X$  
is 
$H_{\infty}(X)= -\log (\max_{x} (\mathrm{P}_X({x})))$, and 
the \emph{average conditional min-entropy}  \cite{DodisORS08} of  RV $X$ given RV $Y \in \mathcal{Y}$ is defined as,
$\tilde{H}_{\infty}(X|Y)= -\log \mathbb{E}_{{y} \leftarrow  \mathcal{Y}}\max_{{x} \in \mathcal{X}}\mathrm{P}_{X|Y}({x}|{y}).$
The \emph{statistical distance} between two RVs $X$ and $Y$ with the same domain $\mc T$ is given by:

\hspace{.25in} ${\rm \Delta}(X,Y)=\frac{1}{2} \sum_{v\in {\cal T}} |\Pr[X=v]-\Pr[Y=v]|.$


We use $U_{\ell}$ to denote a RV that is uniformly distributed over $\{0,1\}^{\ell}$. 
 For an $n$-bit vector $x$, we
write $(x)_{i \cdots j}$ to denote the subvector starting at the 
$i$th bit and ending at the 
$j$th bit, inclusively. 
For a sequence  of variables $W=W_1, \cdots, W_n $, and an index set $A=\{i_1, \cdots, i_t\}, i_j \in [1,\cdots, n]$,  we use 
$W[A]$ to denote the subsequence $W_{i_1}, \cdots, W_{i_t}.$ For a vector $m=(m_0, \cdots, m_{n-1}), m_i $ over some field, let $m(x)$ be the polynomial $\sum_{i=0}^{n-1}m_ix^i$.

\vspace{1mm}
\noindent
{\bf Extractors and hash functions.} 
An extractor distills an almost 
uniform random string from a sample of a random source $W$ that has sufficient entropy, possibly using 
a truly random seed.

\vspace{-.5em}
\begin{definition}[ Strong (average case) randomness extractor~\cite{dodis2006robust}]\label{defn:strext} $E: \{0, 1\}^n\times \{0, 1\}^r\rightarrow \{0, 1\}^\lambda$ is an average case  $(n, \alpha, \lambda,  \epsilon)$-extractor if for any pair of random variables $X, A$ with $X$ over $\{0,1\}^n$ and $\tilde{H}_\infty(X|A)\ge \alpha$, we have  
 $\Delta(E(X, R), A, R; U, A, R)\le \epsilon,$
where $R$ and $U$ are 
uniformly distributed over $\{0,1\}^r$ and 
$\{0, 1\}^\lambda$, respectively. An extractor $E$ is linear if $E(X_1+X_2, R)=E(X_1, R)+E(X_2, R), $ for any $X_1, X_2\in \{0, 1\}^n$ and $R\in \{0, 1\}^r. $
\end{definition}

In above, 
$A$   is the auxiliary variable, and  
if it is null, we have a 
{\em strong $(n, \alpha, \lambda, \epsilon)$-extractor.}
A {\em universal hash   family} gives a randomness extractor 
whose parameters are given by the Leftover Hash Lemma~\cite{impagliazzo1989pseudo}.

\vspace{-.5em}
\begin{definition} [Universal hash family]\label{defn:uhf}
    A family of functions $h:\mathcal{X} \times \mathcal{S} \to \mathcal{Y}$ is called a {\em universal hash family} if $\forall x,y \in \mathcal{X}$ and $x \ne y$, we have:   $\pr[h(x,S)=h(y,S)] \le \frac{1}{|\mathcal{Y}|}$, where the probability is over the uniform choices over $\mathcal{S}$.
\end{definition}
 
\vspace{-1em}
 \begin{lemma}[Generalized Leftover Hash Lemma~\cite{DodisORS08}]
 \label{glhl}
Let \\$h: \mathcal{X} \times \mathcal{S} \rightarrow \{0,1\}^{\ell}$ be a universal hash family. Then for any two  random variables $A \in \mathcal{X}$ and $B \in \mathcal{Y}$, applying $h$ on $A$ can extract a uniform random variable with  length $\ell$ satisfying:   $\Delta(h(A, S), S, B; U_\ell, S, B)\le \frac{1}{2}\sqrt{2^{-\tilde{H}_{\infty}(A|B)}\cdot 2^
\ell}$, where $S$ is chosen uniformly from $\mathcal{S}$. In particular, if $h$ is a universal hash family satisfying $\ell \leq \alpha + 2 - 2 \log{\frac{1}{\epsilon}}$, then $h$ is an  $(n,\alpha,l,\epsilon)$-extractor, where $A$ is an $n$-bit string.
 \end{lemma}
 

\noindent
{\bf Message Authentication Code (MAC)} is a
tuple of algorithms $(\mathsf{gen},\mathsf{mac}, \mathsf{vrfy})$ where $\mathsf{gen}$ generates the key $k\in  \mathbf{K}$,   $\mathsf{mac}$ takes a key $k\in  \mathbf{K}$ and a message $m \in \mathbf{M}$, and  produces a tag $t \in \mathcal{T}$ that will be appended to the message. 
$\mathsf{vrfy}$ algorithm takes a key $k\in  \mathbf{K}$ and pair $(m',t') \in \mathcal{M} \times \mathcal{T}$ and outputs, $\mathsf{accept}$ or $\mathsf{reject}$, indicating the pair is valid or invalid under the  key.

A MAC algorithm is {\em correct} if for any  $m\in \mathcal{M}$, we have,
$\pr[ k  \leftarrow \mathsf{gen}(\lambda),\mathsf{vrfy}(k ,m,\mathsf{mac}(k ,m))= {\mathsf{accept}}] =1.$

A {\em MAC algorithm is  one-time  secure} if  the success chance of the adversary in the following two attacks is bounded by $\delta$.

\noindent
(i) for any message and tag pair $m' \in \mathcal{M}$ and $t'  \in \mathcal{T}$,  $ \pr[ \mathsf{vrfy}(k,m',t')=\mathsf{accept}] \le \delta$;\\
 (ii) for any observed message and tag pair $m\in \mathcal{M}$ and tag pair $t \in \mathcal{T}$, for any  $m' \ne m \in \mathcal{M}$ and $t' \in \mathcal{T}$, we have,
 $  \pr[ \mathsf{vrfy}(k,m',t')=\mathsf{accept}\text{ $|$ }
 m, \mathsf{mac}(k,m)=t] \le \delta$.

 For a more detailed definition see Appendix \ref{dfn:MACAppn}.

\subsection{Fuzzy extractors}
\label{defsfuzzyext}
\remove{
 In this section, we recall definition of fuzzy extractors  and define their security properties {\em reuseability} and {\em strong robustness}.
 }

\vspace{-.5em}
 Fuzzy extractors (information theoretic) were introduced by 
 Dodis et al.  \cite{DodisORS08}.

\remove{
We then  require {\color{blue}strongly} robust security against an adversary with access to $q_e$ and $q_e$ generation and reproduction algorithms.
The outline of our basic  definition (without robustness) is given below.

is the same as the definition of computational-fuzzy extractor
in Canetti et al. \cite{Canetti2020} (Definition 1) except that we replace computational distance by statistical distance  to make it information-theoretically secure 
as defined in the work of Dodis et al. \cite{DodisORS08} (Sections 2.5-4.1).
Below we  state the outline of our new definition of information-theoretic fuzzy extractor.
}
\vspace{-.5em}
\begin{definition}[Fuzzy extractor]
\label{def:fuzzyext} 
Let $\mathcal{W}$ be a family of probability distributions over $\mathcal{M}$. An  $(\mathcal{M},\mathcal{W},\xi, t, \epsilon, \sigma)$-fuzzy extractor FE is a pair of   randomized  algorithms ($Gen$, $Rep$) as follows.
\begin{itemize}
\item[i.] 
$Gen$: $\mathcal{M} \rightarrow \{0,1\}^{\xi} \times \mathcal{P}$ 
takes a source sample $w \in  \mathcal{M}$, and outputs a key $r \in \{0,1\}^{\xi}$ and a public {\em helper string} $p \in \mathcal{P}$. 

\item[ii.] 
$Rep$ : $\mathcal{M} \times \mathcal{P} \rightarrow  \{0,1\}^{\xi} $  takes 
a sample $w' \in \mathcal{M}$ and the helper  string $p \in \mathcal{P}$ as input, and outputs a key $r$.  
\end{itemize} 
A fuzzy extractor must satisfy the following properties.

\noindent
1. {\em $\epsilon$-correctness.} \quad An FE 
is $\epsilon$-correct
if for any $w'$ satisfying
 $d(w,w') \le t$ where $d()$ is a distance function, and  $(r, p)\leftarrow Gen(w)$, we have  $Pr[Rep(w', p) = r] \geq 1 - \epsilon$, where the probability is over the  randomness of $Gen$ and $Rep$. 

\noindent
2. {\em $\sigma$-security.} \quad For any distribution  $W \in \mathcal{W}$,  the key $R$ is close to uniform, conditioned on $P$. 
That is, 
for any pair, \\
$(R, P)  \leftarrow Gen(W)$, we have, ${\rm \Delta}((R, P); (U_{\xi}, P)) \le  \sigma$.

\end{definition}

\vspace{.05in} \noindent \textbf{Reusability and robustness. 
}\quad 
 The following 
 definition of reusability of FE is from \cite{Canetti2020}.  The experiment on the LHS of Figure~\ref{fig:gamesreusrob} describes security of $Gen$ as a game between a challenger and an adversary. The challenger provides $\eta$ outputs of the $Gen $ algorithm  on correlated samples to the adversary, after which the adversary must distinguish between a random string and an output key of $Gen$ that has not been revealed. 
The experiment on the RHS of Figure~\ref{fig:gamesreusrob} is defined similarly with the same information provided to the adversary, whose goal is to modify the helper data on a new sample from the same source without being detected.

\vspace{-1em}
\begin{figure}[!ht]
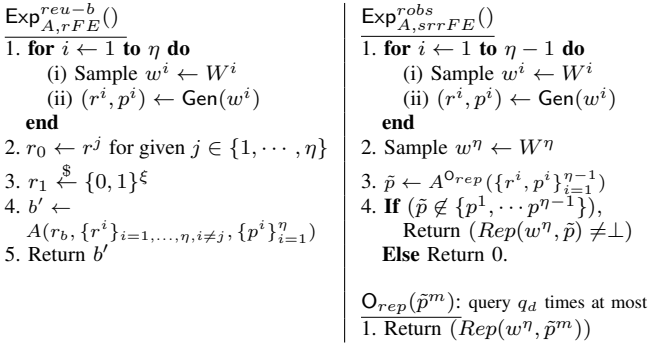

\footnotesize
   \begin{center}
\begin{tabular}{  p{4.3cm} |  p{4cm}}
      \underline{$\mathsf{Exp}^{reu-b}_{A,rFE}()$ } &  \underline{$\mathsf{Exp}^{robs}_{A,srrFE}()$}\\
     1. {\bf for} {$i \gets 1$ \KwTo $\eta$} {\bf do} &  1. {\bf for} {$i \gets 1$ \KwTo $\eta - 1$} {\bf do}\\
      \qquad (i) Sample $w^{i}\gets W^{i}$ & \qquad (i) Sample $w^{i}\gets W^{i}$\\
       \qquad (ii) $(r^{i} , p^{i} )\gets \ms{Gen}(w^{i})$ & \qquad (ii) $(r^{i} , p^{i} )\gets \ms{Gen}(w^{i})$ \\
      \quad  {\bf end}& \quad  {\bf end}\\
      2. $r_0 \gets r^j$ for given $j\in\{1,\cdots,\eta\}$& 2.  Sample $w^\eta\gets W^\eta$  \\ 3. $r_1 \stackrel{\$}{\gets} \{0, 1\}^{\xi} $ & 3. {\scriptsize$\tilde{p} \gets A^{\ms O_{rep}}(\{r^i,p^i\}_{i=1}^{\eta -1})$}   \\ 
      4. {\scriptsize$b' \gets$} & 4. {\bf If} $(\tilde{p} \not\in \{p^1,\cdots p^{\eta -1}\}),  $   \\
       \quad {\scriptsize $A(r_b,\{r^i\}_{i=1,...,\eta, i\neq j}, \{p^i\}_{i=1}^{\eta})$} &\qquad Return $(Rep(w^\eta,\tilde{p}) \neq \perp )$  \\ 5. Return $b'$  & \quad  {\bf Else} Return 0. \\
        &  \\
                &  \underline{$\ms O_{rep}(\tilde{p}^m)$:} {\scriptsize query $q_d$ times at most} \\
                &  1. Return $(Rep(w^\eta,\tilde{p}^m))$  
\end{tabular}
\end{center}
\vspace{-1em}
    \caption{\footnotesize Security experiments for reusability (LHS) and robustness (RHS) for rFE and srrFE, respectively.
    %
  }
    \label{fig:gamesreusrob}
   \vspace{-1.5em}
\end{figure}

\begin{definition}[Reusable fuzzy extractors \cite{Canetti2020}]\label{def:rfuzzyext}
  Let $\mathcal{W}$ be a family of 
  distributions over $\mathcal{M}$, and {\sf rFE}=$(Gen,Rep)$ 
denote an
 $(\mathcal{M},\mathcal{W},\xi, t, \epsilon, \sigma)$-fuzzy extractor as given in Definition \ref{def:fuzzyext}. 
 The FE is {\em $(\eta, \sigma_r)$-reusable} 
if for  any  $\eta$ correlated RVs   $(W^1,...,W^{\eta}), W^j\in \mc W, \forall j \in \{1,...,\eta\}$, and  any computationally  unbounded adversary $A$ in the experiment $\mathsf{Exp}^{reu}_{A,rFE}$ in Figure~\ref{fig:gamesreusrob},
\remove{ if the following 
 property holds.  
 Let $(W^1,...,W^{\eta})$ be 
 $\eta$ correlated RVs, 
 where $W^j\in \mc W, \forall j \in \{1,...,\eta\}$.  For any computationally  unbounded adversary $A$ 
 in experiment $\mathsf{Exp}^{reu}_{A,rFE}$ defined in the left side of Figure~\ref{fig:gamesreusrob},
 } 
 the following advantage 
 is bounded by $\sigma_r$: 
\vspace{-.5em}
\begin{scriptsize}
\begin{eqnarray*}
Adv_A (\mathsf{Exp}^{reu}_{A,rFE}) 
= |\Pr[\mathsf{Exp}^{reu-0}_{A,rFE}()=1] - 
\Pr[\mathsf{Exp}^{reu-1}_{A,rFE}()=1]| \leq \sigma_r &&.
\vspace{-2.5em}
\end{eqnarray*}
\end{scriptsize}
\remove{An $(\mathcal{M},\mathcal{W},\xi, t, \epsilon, \sigma)$-fuzzy extractor is {\em $(\eta, \sigma_r)$-reusable} if for all  adversary $A$ 
and for all $j = 1, ...,\eta$, the advantage is at most $\sigma_r$}.
\end{definition}
\remove{\noindent
{\bf Robustness}
Robustness of fuzzy extractors is defined against an adversary who tampers with $p$ that is input to the reconstruction  algorithm.
Robustness was first defined by Boyen et al.~\cite{boyen2005secure} for an adversary who has access to $p$ only  (referred to as S2 in Section \ref{sec:intro}), and  was strengthened by 
Dodis et al.~\cite{dodis2006robust}  to the case when the adversary has  access to the pair $(p,r)$, and attempts to modify $p$  (referred to as S2 in Section \ref{sec:intro}).
}
\vspace{-1.5em}
We also consider {\em strongly robust and reusable FE (srrFE for short)} defined as follows.


\vspace{-.5em}
\begin{definition}[Strongly robust of reusable fuzzy extractor]\label{def:srrfuzzyext}
\remove{
Let $\mathcal{W}$ be a family of probability distribution over $\mathcal{M}$, 
and let  {\sf srrFE}=$(Gen,Rep)$ 
 be an 
 $(\mathcal{M},\mathcal{W},\xi, t, \epsilon, \sigma)$-fuzzy extractor that is $(\eta, \sigma_r)$-reusable. 
 Let $(W^1,...,W^{\eta}, W')$ be 
 $\eta+ 1$ correlated RVs, where $W', W^j\in \mc W, \forall j \in \{1,...,\eta\}$.
 }

 Consider the  setting of Definition \ref{def:rfuzzyext} and an $(\eta, \sigma_r)$-reusable FE as defined above such that {\footnotesize $d(W^i,W^\eta) \le t, \forall i \in \{1,\cdots,\eta-1\}$}. 
We say the FE is a  {\sf srrFE}, that  is 
 {\em $(q_d,\delta_r)$-strongly robust}, 
if the success probability of  any computationally  unbounded adversary $A$ 
 in the robustness  experiment $\mathsf{Exp}^{robs}_{A,srrFE}$ defined in the the right side of Figure \ref{fig:gamesreusrob}   
 is bounded by $\delta_r$: 

\vspace{-1em}
{\footnotesize
\begin{align*}
    Adv_{A}(\mathsf{Exp}^{robs}_{A,srrFE})=\Pr[\mathsf{Exp}^{robs}_{A,srrFE}()=1] \leq \delta_r 
    \vspace{-.5em}
\end{align*}
}
 The term ``strong" refers to arbitrary correlated queries (random variables) that can be chosen by  the adversary.

\remove{Game $Exp_{FE_{srr},\ms D}^{rob}(1^\lambda)$: We define the following game for all $j = 1, ..., \eta$:
\begin{itemize}
    \item \textbf{Sampling:} Challenger samples $w^j\gets W^j$ and $w'\gets W'$   
    \item \textbf{Generation:} Challenger computes $(r^j , p^j )\gets \ms{Gen}(w^j)$ and returns $(r^j, p^j )_{j=1}^n$ to the adversary $\ms D$.
    \item \textbf{Reproduction oracle queries:} Adversary $\ms D$ may adaptively make at most $q_d$ reproduction oracle queries of the form $\tilde{p}^m$, $m \in \{1,\cdots,q_d\}$, to $\mathcal{C}$. Challenger $\mathcal{C}$ runs $Rep(w',{\tilde{p}}^m)$ and returns its output to $\ms D$, $\forall m \in \{1,\cdots,q_d\}$.
    \item \textbf{Forgery test:} $\ms D$ submits its forgery $\tilde{p}$ to $\mathcal{C}$. A wins if $\tilde{p} \not\in \{p^1,\cdots p^\eta\}$ and $Rep(w',\tilde{p}) \neq \perp$. The experiment outputs $1$ if $\ms D$ wins and $0$ otherwise.
\end{itemize}
}

\end{definition}


\section{Construction~\ref{const1} \&~\ref{const2:cca}: IT-secure 
fuzzy extractor for structured sources}
\label{construction:const1}
\vspace{-.5em}
In this section we present the construction of an rFE for a {\em structured source}, and then extend it to an srrFE by using a MAC.
The constructions 
are for 
for {\bf $(\alpha, m, N)$-sources} that are defined below.
\vspace{-.5em}
\begin{definition}[\bf $(\alpha, m, N)$-source] \label{defn:condsource} 
Consider  a source $W=W_1, \cdots, W_n$  that consists of strings of length $n$  over an alphabet ${\cal Z}$.   For parameters $\alpha$, $m$ and $ N$, we say $W$  is an {\em $(\alpha, m, N)$-source} if 
for any two random subsets $A, B\subseteq [1,\cdots, n]$, with cardinality $m$ and $N$, respectively, we have   $\tilde{H}_\infty(W[A] \mid W[B], A, B)\ge \alpha$,
where the probability is over the randomness of $W$ and index sets $A, B.$
\vspace{-.5em}
\end{definition}
 Intuitively the requirement is that entropy of a random subsample  of size $m$ is high, when any 
substring of length $N$  of the sample is seen. 
The  $\alpha$-entropy $k$-samples sources defined in \cite{Canetti2020}, are  less restrictive class of  $(\alpha, k, N)$-sources when $N=0$.
Our constructions are for binary alphabets,  ${\cal Z}=\{0, 1\}. $

\vspace{.05in} 
\begin{construction}\label{const1} (rFE)   Let $E$ be an average $(m, \alpha, \nu,   \epsilon)$-extractor,  
and $W=W_1,W_2,...,W_n$ be an $(\alpha, m, N)$-source with $W_i$ over alphabet $\{0, 1\}$. 
The $Gen$ and $Rep$ algorithms in Figure \ref{fig:const1} 
describe the  fuzzy extractor generation $Gen$  and reproduction procedures  $Rep$, respectively,  for the  source. 
\end{construction}
\vspace{-1em}
\begin{figure}[!ht]
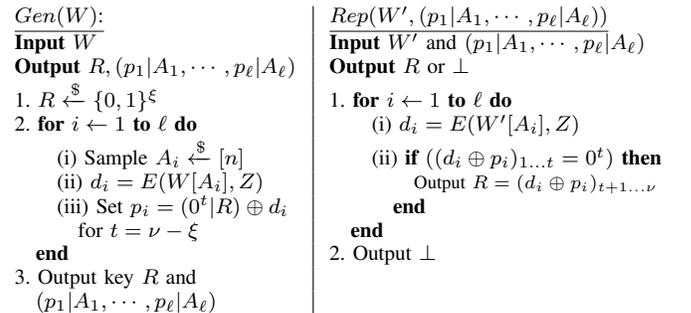

\footnotesize
   \begin{center}
\begin{tabular}{  p{3.75cm} |  p{4.5cm}}
       \underline{$ Gen(W)$:} &  \underline{$ Rep(W',(p_1|A_1, \cdots, p_\ell|A_\ell))$}\\ 
     {\bf Input} $W$ &   {\bf Input} $W'$ and $(p_1|A_1, \cdots, p_\ell|A_\ell)$ \\
     {\bf Output} $R,(p_1|A_1,\cdots, p_\ell|A_\ell)$ &   {\bf Output} $R$ or $\perp$ \\
     1. $R\xleftarrow{\$} \{0,1\}^{\mathcal{\xi}}$ &   1. {\bf for} {$i \gets 1$ \KwTo $\ell$} {\bf do}\\
      2. {\bf for} {$i \gets 1$ \KwTo $\ell$} {\bf do}&  \qquad (i) $d_i = E(W'[A_i], Z)$\\
       \qquad (i)  Sample $A_i \xleftarrow{\$} [n]$  & \qquad (ii) {\bf if} {$((d_i \oplus p_i)_{1...t}=0^{t})$} {\bf then} \\ 
       \qquad (ii) $d_i = E(W[A_i], Z)$
       & \qquad  \qquad  {\scriptsize Output  $R=(d_i \oplus p_i)_{t+1...\nu}$} \\
       \qquad (iii) Set $p_i = (0^{t}|R) \oplus d_i$ & \quad  \qquad   {\bf end} \\
      \qquad \quad  for $t=\nu-\xi$ & \quad  {\bf end}\\
      \quad {\bf end}    &  2. Output 
      $\perp$ \\
     3. Output key $R$ and & \\
       \quad $(p_1|A_1,\cdots, p_\ell|A_\ell)$ & \quad \qquad     
\end{tabular}
\end{center}
\vspace{-1em}
    \caption{\footnotesize 
    Construction~\ref{const1}. 
     $\ell, t, \lambda, \xi, n, m, L$ are public parameters. 
     The 
     randomness $Z$ is shared through an authenticated channel.
   $(\rho)_{i\cdots j}$ denotes the substring starting at bit $i$ and ending at 
   bit $j$ of $\rho.$ 
  %
  In line $(i)$ of $Gen(W)$, 
  $A_i \xleftarrow{\$} [n]$ denotes choosing
  $A_i=\{i_i,\cdots,i_m\},  i_j \xleftarrow{\$} [n], 
  \forall j \in  [1,m] 
  $.}.
    \label{fig:const1}
    \vspace{-1.6em}
\end{figure}

 The construction is inspired by the {\em sample-then-lock} approach of Canetti et al. \cite{Canetti2020}, outlined in section \ref{sec:intro}.

\remove{
\subsubsection{Security analysis of fuzzy extractor Construction~\ref{const1}}
\label{securityfuzzyextconst1}
%
%
In the following we prove the construction is a reusable FE, satisfying Definition \ref{def:rfuzzyext}.
The construction is reusable in the sense of (R1.2, R2.2).
} 
\begin{theorem}[Reusable fuzzy extractor]
\label{THM:FUZZYEXTOTCONST1}  In $Gen$ algorithm in Figure \ref{fig:const1}, let $E$ be instantiated 
by $H$, a 
 universal hash family. 
 Let $\xi$ denote the length of the extracted key. For a chosen value $\ell$, and  $\epsilon' $, $\sigma$ 
that satisfy {\footnotesize $\xi \leq \alpha + 2 - 2 \cdot \log(\frac{2 \ell}{\sigma}) - t$}, the $(Gen, Rep)$ procedures in  Figure~\ref{fig:const1} correspond to an $(\eta,\sigma)$-reusable $(\mathcal{V}^n,\mathcal{W},\xi,t',\epsilon',\sigma)$-fuzzy extractor,
 where
\\{\footnotesize $(1-(1-\frac{t'}{n-m})^m)^\ell + \ell \cdot 2^{-t} \leq \epsilon'$} 
and $\ell \eta m<N$.
\end{theorem}
\begin{proof}
The following is a proof outline. The complete proof is in Appendix~\ref{pf:THM:FUZZYEXTOTCONST1}. 
\remove{an outline of the proof. 
To prove that the construction satisfies 
Definition~\ref{def:rfuzzyext}, we prove 
$\epsilon'$-correctness according to Definition~\ref{def:fuzzyext} and then reusability according to Definition~\ref{def:rfuzzyext}.
}

 \textit{
Correctness.} 
We show that 
$Rep$ algorithm fails with probability at most $(1-\epsilon')$. 
 The $Rep$ algorithm fails in two cases. 
First, none of the $\ell$ subsamples of length $m$ in $W'$ matches the corresponding subsample of $W$. 
We derive expressions that bounds and shows how to choose $\ell$ to have 
 at least one matched subsample of length $m$  in $w$ and $w'$ with probability at least  $(1-\epsilon')$.\\
Second, there is a match in step $(ii)$ of $Rep$ because of collision in the universal hashing. More specifically,  the hash value of a subsample in $W$  is equal to the hash value of the corresponding subsample in $W'$ {\em in the first $t$ bits} (the rest is padded with   random key).

 This collision in step $(ii)$ of $Rep$ in figure~\ref{fig:const1}, occurs with probability $(\frac{1}{2^{t}})$ when $E$ is 
 a randomly chosen 
 from a universal hash family. 
 In this case $Rep$ will output a key $R'\neq R$ because  $Rep$ algorithm ends as soon as the first  decryption succeeds.

\textit{Reusability.}  
We first prove for 
$\eta=1$, and  then  extend it to 
$\eta > 1$. 
{\bf 
$\eta=1$}. We need to prove that ${\rm \Delta}((R, P), (U_{\xi}; P)) \le  \sigma$, where $P$ is the public strings.

 In step $(i)$ of  $Gen()$ algorithm  in figure~\ref{fig:const1}, 
$A_i$ is randomly sampled, and 
in step $(ii)$, $E$ is implemented by $H$.

Since $W$ is a $(\alpha, m, N)$-
source, we have $\tilde{H}_\infty(W[A] \mid  W[B], A, B)\ge \alpha$, where $A$ 
is a  random subset of $[n]$, and $B\subset [n]$ of size   $N$ and disjoint from $A$.
Since $H$ is 
universal hash family with output length $\nu$, from lemma~\ref{glhl},
$H$ is an $(m,\alpha,\nu,\frac{\sigma}{2\ell})$-extractor as long as
 $\nu \leq \alpha + 2 - 2\log(\frac{2\ell}{\sigma})$.  
Since $\nu = \xi + t$, we obtain an upper bound   on $\xi$ for
$H$ to be an $(m,\alpha,\nu,\frac{\sigma}{2\ell})$-extractor, and we have $\Delta(R, {\bf p}, Z, {\bf A}; U, {\bf p}, Z, {\bf A})\le \sigma.$
\remove{
\begin{equation}
\label{sdadvconst1}
\Delta(R, {\bf p}, Z, {\bf A}; U, {\bf p}, Z, {\bf A})\le \sigma.
\end{equation}

Since $H:\mathcal{X} \times \mathcal{S} \rightarrow \{0,1\}^\nu$ is a  
universal hash family, from lemma~\ref{glhl}, if $\nu \leq \alpha + 2 - 2\log(\frac{2\ell}{\sigma})$, $H$ is an $(m,\alpha,\nu,\frac{\sigma}{2\ell})$-extractor.  Now, $\nu = \xi + t$. Thus, if $\xi \leq \alpha + 2 - 2\log(\frac{2\ell}{\sigma}) - t$, $H$ is an $(m,\alpha,\nu,\frac{\sigma}{2\ell})$-extractor. 
}


\remove{Since $H$ is an $(m,\alpha,\nu,\frac{\sigma}{2\ell})$-extractor, from lemma~\ref{le: dist}, \ref{le: d} in  Appendix~\ref{pf:THM:FUZZYEXTOTCONST1}, if $\xi \leq \alpha + 2 - 2\log(\frac{2\ell}{\sigma}) - t$ and $\ell m<N$, considering $S=R$ in lemma~\ref{le: d} in Appendix~\ref{pf:THM:FUZZYEXTOTCONST1},   we have
\begin{equation}
\label{sdadvconst1}
\Delta(R, {\bf p}, Z, {\bf A}; U, {\bf p}, Z, {\bf A})\le \sigma.
\end{equation}
}
Therefore, $(Gen,Rep)$ is an srrFE with 
$(1-(1-\frac{t'}{n-m})^m)^\ell + \ell \cdot 2^{-t} \leq \epsilon'$, $\ell m<N$ and $\mathcal{V} \in \{0,1\}$ 
and the extracted key length $\xi \leq \alpha + 2 - 2\log(\frac{2\ell}{\sigma}) - t$.

 {\bf Case $\eta>1$}. Similar to the case $\eta=1$ and taking into account property of $(\alpha,m,N)$-source that ensures  the entropy of the new samples remain the same. 
\end{proof}

\vspace{1mm}
 \remove{
 \subsection{Construction 3: strongly robust and reusable fuzzy extractor}
\label{construction:const2}
\subsection{Adding robustness}
Construction~\ref{const1} is an IT-secure reusable FE. We extend this construction to a strongly robust, reusable, IT-secure FE  (in the standard model). To our knowledge 
}
\vspace{-.5em}
\subsubsection{A strongly robust and 
reusable FE}  
  Construction 2 is an srrFE and is in 
   {\em CRS model} which assumes that sender and receiver share a public random string that
is  independent of the source, and is  used  to specify 
$A_1, \cdots, A_\ell$ and $Z$.  
Robustness  is obtained by  using an information theoretic MAC with a special property, together with Construction 1.
%
This is the first  IT-secure srrFE without assuming random oracle, and is for structured sources.

\begin{construction}\label{const2:cca}  \remove{Let $E$ be an average $(m, \alpha, \nu,   \epsilon)$-extractor, and 
$W=W_1,...,W_n$ denote an $(\alpha, m, N)$-source, where $W_i \in \{0, 1\}$.} 
Consider the same source as Construction 1, and $E$ be an average $(m, \alpha, \nu,   \epsilon)$-extractor.
Figure~\ref{fig:const12}  describes  the 
$Gen$ 
and 
$Rep$ algorithms of our srrFE.
\remove{Figure~\ref{fig:gen1} and~\ref{fig:gen2}(resp. Figure~\ref{fig:rep1} and ~\ref{fig:rep2}) 
in  Appendix~\ref{appendix_fig_ccaITconst} 
depict a pictorial representation of our fuzzy extractor generation (resp. reproduction) algorithm, where the encoding algorithm for $p_1|\cdots|p_\ell$ is described later in this section. }   
\end{construction}
\remove{{\bf CRS Model. }  We assume that sender and receiver have access to a common random string that
is used  to specify 
$A_1, \cdots, A_\ell$ and $Z$. 
The CRS is  independent of the source.  The same 
is  used even in the  computational robustly reusable fuzzy extractor construction in~\cite{WenLAsiaCr18}.
} 

\remove{\noindent
{\bf Bit-string Encoding.} 
\quad 
To encode a bit string into a string over 
$GF(2^\lambda)$,
we use a 
primitive element $\omega$ of $GF(2^\lambda)$ and 
encode  
the bit string
$a_0a_1\cdots a_{\lambda-1}$ as $\sum_{i=0}^{\lambda-1}a_i\omega^i$. An incomplete block is appropriately padded with 
zeros.
}

\vspace{-1.5em}
\begin{figure}[!ht]
\scriptsize
   \begin{center}
\begin{tabular}{  p{3.85cm} |  p{4.5cm}}
      \underline{$Gen(W)$} &  \underline{$Rep(W',(p_1, \cdots, p_\ell, T))$}\\
     {\bf Input} $W$ &   {\bf Input} $W'$ and $(p_1, \cdots, p_\ell, T)$ \\
     {\bf Output} $R,(p_1,\cdots, p_\ell,T)$ &   {\bf Output} $R$ or $\perp$ \\      
     1. $R\xleftarrow{\$} \{0,1\}^{\mathcal{\xi}}$, $R_1\xleftarrow{\$} \{0,1\}^{2\mathcal{\lambda}}$&  1. {\bf for} {$i \gets 1$ \KwTo $\ell$} {\bf do}\\
      2. {\bf for} {$i \gets 1$ \KwTo $\ell$} {\bf do}&  \qquad (i) $d_i = E(W'[A_i], Z)$\\
       \qquad (i)  Sample a  random subset  &  \qquad (ii) {\bf if} {$((d_i \oplus p_i)_{1...t}=0^{t})$} {\bf then} \\ $\empty$ \hspace{.25in} 
       $A_i=\{i_1, \cdots, i_m\}$ from &\qquad \qquad (a) Set $\rho=(d_i \oplus p_i)_{t+1...\nu}$ \\
       $\empty$ \hspace{.25in} $[n]$ 
       & \qquad \qquad  (b) Set $R = (\rho)_{t+1...t+\xi}$, \\
      \qquad (ii) $d_i = E(W[A_i], Z)$&  \qquad \qquad \quad  $R_1 = (\rho)_{\nu-2\lambda+1...\nu}$,  \\
      \qquad (iii) Set $p_i = ((0^{t}|R|R_1) \oplus d_i)$ &  \qquad \qquad  (c) $T'=$  \\
      \qquad \qquad for $t=\nu-\xi-2\lambda$ & \qquad \qquad \quad $Eval((p_1,\cdots, p_\ell),R_1,L)$ \qquad \qquad \\
      \quad {\bf end}&  \qquad \qquad (d) {\bf if} $T'=T$, {\bf output} key $R$;  \\ 
      4. Let $L=\lceil \ell \nu/\lambda \rceil+4$, and &  \quad \qquad \qquad   {\bf otherwise}, continue  \\
      \quad $p=(p_1,\cdots, p_\ell)$ & \quad \qquad   {\bf end} \\
       5. $T=Eval(p,R_1,L)$ & \quad  {\bf end}   \\ 
      6. Output key $R$ and ciphertext $(p_1,\cdots, p_\ell, T)$ &  2. Output $\perp$ 
\end{tabular}
\end{center}
\vspace{-1.2em}
    \caption{\footnotesize 
     srrFE  $Gen$ and $Rep$ algorithms.  $\ell, t, \lambda, \xi, n, m, L$ are system public parameters. $E$ is an average $(m, \alpha, \nu, \epsilon)-$extractor. The 
    subsets $A_i$ and $Z$ are 
    obtained from CRS. 
    $Eval(\cdot)$ is  defined in Algorithm~\ref{alg:compufuzzygeneration11}.}
    \label{fig:const12}
    \vspace{-1.5em}
\end{figure}
\vspace{-.7em}
  {\setlength{\algoheightrule}{0pt}
\setlength{\algotitleheightrule}{0pt}
\SetAlgoNoLine%
 	\begin{algorithm}[!ht]
 	\footnotesize
 		\SetAlgoLined
 		\DontPrintSemicolon
            \rule{.4\textwidth}{0.95pt}
            
 		\SetKw{KwBy}{by}
 		\SetKwBlock{Beginn}{beginn}{ende} 
   		1. Encode $p$ to vector ${\bf m}$ of length $L-4$ in $GF(2^\lambda)$  \\
 		2. Parse $R_1=x|y$ for $x, y \in \{0, 1\}^\lambda$  \\
        3. Compute $T=x^L+x^2{\bf m}(x)+xy$ for ${\bf m}(x)=\sum_{i=0}^{L-5}m_ix^i.$\\
        4. Return $T$ \\ 
   \vspace{-.9em}
 		\parbox{\linewidth}{\caption{\footnotesize $T \leftarrow Eval(p,R_1,L)$}\label{alg:compufuzzygeneration11}}  
  \vspace{-1.7em}
 	\end{algorithm}
 }




\remove{In our construction, this CRS assumption is  important. Otherwise,  the adversary might   select  some index set $A_i$ with low entropy for $w[A_i]$ and hence   correctly guess $w[A_i]$ with high probability, and then  generate 
 $P$ that is  acceptable by $Rep$ algorithm.
CRS model prevents attacker  from choosing a weak sample.}


\noindent
{\em  Providing robustness.}
To provide  robustness for the rFE, we use a part of  the extracted key to compute an authentication tag  using a one-time IT-secure MAC, and append it to $P=p\parallel tag$.
%
We note 
that  if an adversary modifies $p$ to $\hat{p}$, because of the linearity of 
``one-time-pad''  
the key $(R\parallel  R_1)$ will be shifted by a value $\delta=p+\hat{p}$ that is known  to the adversary.
We use a  MAC construction that provides  security against a known  key shift.

For key $(x, y)$, the 
tag function is, $T(x, y, {\bf m})= x^L+x^2{\bf m}(x)+xy$, where ${\bf m}$ is a vector over $GF(2^\lambda)$ of length at most $L-5$ (so  ${\bf m}(t)$ is a polynomial of degree at most $L-5$; see notations in  Section \ref{sec:pre}). 
The verification algorithm for $(m',t')$ is by computing the tag function for $m'$, and comparing the result with $t'$.

In Appendix~\ref{pf:lemmamac},
we prove that the MAC  provides one-time security.

\begin{theorem}[Robust and reusable FE]
\label{THM:FUZZYEXTOTCONST2} Let $\mathcal{W}$ be a family of $(\alpha, m, N)$-sources and $E$ be  an average  $(m, \alpha, \nu, \epsilon)$-randomness linear extractor for source $W$. 
 Let $E$ be implemented by $H$, a 
universal hash family. Fix $\ell$ and let $\xi$ be the length of the extracted key. If adversary makes at most $q_e$ generation queries and $q_d$ reproduction queries, then for any $\delta$, $\epsilon'$, $\sigma$ satisfying $\xi \leq \alpha + 2 - 2 \cdot \log(\frac{4 \ell}{\sigma}) - t - 2\lambda$, $(Gen, Rep)$ described in Construction~\ref{const2:cca}  is a $(q_d,\delta)$-strongly robust, $(q_e+1,\sigma)$-reusable $(\mathcal{V}^n,\mathcal{W},\xi,t',\epsilon',\sigma)$-fuzzy extractor, where  
$(1-(1-\frac{t'}{n-m})^m)^\ell + \ell \cdot 2^{-t} \cdot L\cdot 2^{-\lambda} \leq \epsilon'$, $(q_e+1)\ell m<N$,  $\delta=(q_d+q_e)\ell\epsilon+q_d2^{-\lambda}\ell (L+1)$,  $(q_e+q_d)\ell m<N$ and $L=\lceil\ell \nu/\lambda\rceil+4$. 
\end{theorem}

    The proof of this theorem is in Appendix~\ref{pf:THM:FUZZYEXTOTCONST2}.

\section{Concluding remarks}
We proposed an rFE and an srrFE, both with information theoretic security for a structured source.  Both FEs provide  security for strong notion of reusuability. The srrFE is the only known strongly robust and reusable IT-secure FE that uses CRS model.
\remove{
Construction  2  answers an open question (Canetti et al. \cite{Canetti2020}) in information theoretic setting.  All constructions are secure against an adversary with access to quantum computer.} 
Extending our FEs to 
more general sources are interesting open questions. 

{\footnotesize
\bibliographystyle{IEEEtran}
\bibliography{main}}
\normalsize
\appendix
\subsection{\textbf{Message Authentication Code}}\label{dfn:MACAppn}
\begin{definition}[one-time MAC] An  $(\mathcal{M},\mathcal{K},\mathcal{T},\delta)$-one-time MAC is a triple of polynomial-time algorithms $(\mathsf{gen},\mathsf{mac}, \mathsf{vrfy})$ where $\mathsf{gen}(\lambda): 1^\lambda \to  \mathcal{K}$ gives a key $k \in \mathcal{K}$ on input security parameter $\lambda$,  $\mathsf{mac}: \mathcal{K} \times \mathcal{M} \to \mathcal{T}$ takes the key $k$ and a message $m \in \mathbf{M}$ as input and outputs a tag $t \in \mathcal{T}$, $\mathsf{vrfy}: \mathcal{K} \times \mathcal{M} \times \mathcal{T} \to \{\mathsf{accept},\mathsf{reject}\}$ outputs either $\mathsf{accept}$ or $\mathsf{reject}$ on input: the key $k$, a message $m  \in \mathcal{M}$  and a tag  $t \in \mathcal{T}$.  
    
    The {\em correctness} property requires that: for any choice of $m\in \mathcal{M}$, we have 
 {   \footnotesize
 \begin{equation} 
\pr[ k  \leftarrow \mathsf{gen}(\lambda),\mathsf{vrfy}(k ,m,\mathsf{mac}(k ,m))= {\mathsf{accept}}] =1
\vspace{-.3em}
 \end{equation}
 }
 The unforgeability property requires that for any $k \leftarrow \mathsf{gen}(\lambda)$, the following two conditions are satisfied:
 
 (i) for any message and tag pair $m' \in \mathcal{M}$ and $t'  \in \mathcal{T}$,  
 \begin{equation}
  \pr[ \mathsf{vrfy}(k,m',t')=\mathsf{accept}] \le \delta, \\
  \vspace{-.3em}
 \end{equation}
 (ii) for any observed message and tag pair $m\in \mathcal{M}$ and tag pair $t \in \mathcal{T}$, for any adversary's choice of $m' \ne m \in \mathcal{M}$ and $t' \in \mathcal{T}$,
 \vspace{-.4em}
{\small
\begin{equation}
  \pr[ \mathsf{vrfy}(k,m',t')=\mathsf{accept}\text{ $|$ }
 m, \mathsf{mac}(k,m)=t] \le \delta, \\
 \vspace{-.3em}
 \end{equation}
 }.
\end{definition}

\subsection{\textbf{Security proof of  the MAC}}\label{pf:lemmamac}
\noindent The following lemma shows that the tag function has one-time authentication property.
\vspace{-.5em}
\begin{lemma}
Let $L=3$ mod 4 and  $m$ be an arbitrary but given   vector of length at most $L-5$ over $GF(2^\lambda).$  Let $x, y$ be uniformly random over $GF(2^\lambda).$ Then, given $T=x^L+x^2m(x)+xy$, the following holds 
{\footnotesize
\begin{equation}
(x+\delta_1)^L+(x+\delta_1)^2 m'(x+\delta_1)+(x+\delta_1)(y+\delta_2)=T' \label{eq: mm2'}
\end{equation}
}
 with probability at most $L2^{-\lambda}$, where $T', \delta_1, \delta_2\in GF(2^\lambda)$ and  $m'$  a vector over $GF(2^\lambda)$ of length at most $L-5$ with $m'\ne m$,  are arbitrary but all deterministic in $T$ and  the probability is over the choices of $(x, y)$.    \label{le: forge}
\end{lemma}

\begin{proof}
We write {\em $(m', T)$ valid} to indicate the event  that  Eq. (\ref{eq: mm2'}) is satisfied. Then, 
given $T$, the probability that  Eq. (\ref{eq: mm2'}) holds,  is  
\begin{align} 
\nonumber 
&\mathbb{E}_{T}\Big{(}P_{xy}[(m', T') \mbox { valid}  \mid  T]\Big{)} \\ 
\nonumber 
\leq &\sum_{a \in GF(2^\lambda)}P_{xy}[(m',T') \mbox{ valid},  T=a] \\ 
\nonumber 
\leq &\sum_{a}P_{xy}[(m',T') \mbox{ valid},  T=a, x=0] \\
\nonumber 
&+\sum_{a}P_{xy}[(m',T')  \mbox{ valid}, T=a, x \ne 0] \\ 
\leq &2^{-\lambda} + \sum_{a}P_{xy}[(m',T')  \mbox{ valid}, T=a, x \ne 0]  \label{eq: JSQ-0} 
  \end{align} 
Next, we bound  $\sum_{a}P_{xy}[(m',T')  \mbox{ valid}, T=a, x \ne 0]. $
  To  bound this, we study two cases:  case $\delta_1=0$ and case $\delta_1\ne 0. $ In the following, we  use the fact that  $T', m', \delta_1, \delta_2$ are determined by $T$.  Therefore, given $T=a$, $(T', m', \delta_1, \delta_2)$ are all fixed. 

{\bf Case $\delta_1=0$. } In this case, since $m'\ne m$, $T'-T=x^2(m'(x)-m(x))+x\delta_2$ is a non-zero polynomial of degree at most $L-3$ (when fix $T=a$). Hence, given $T,T'$, there are at most $L-3$ possible $x$ that satisfy this. Moreover, for any  $a\in GF(2^\lambda)$ and non-zero $x$, there exists a unique $y$ such that  $(x, y)$ results in  $T=a$. As a result, given $T=a\wedge x\ne 0,$
$x$ is uniformly random over $GF(2^\lambda)-\{0\}.$  Keeping these facts in mind,   we obtain  
{\footnotesize
\begin{align} 
\nonumber 
&P_{xy}[(m',T')  \mbox{ valid} \mid  x \ne 0, T=a, \delta_1=0]\cdot 
\Pr(x \ne 0, T=a, \delta_1=0) \\ 
\nonumber 
\leq & P_{xy}\big[(x^2(m'(x)-m(x))+x\delta_2=T'-a)  \mid  x \ne 0, T=a, \delta_1=0\big]\cdot \\\nonumber
&\qquad \Pr(x \ne 0, T=a, \delta_1=0) \\\nonumber  
\leq & P_{xy}\big[(x^2(m'(x)-m(x))+x\delta_2=T'-a)  \mid  x \ne 0, T=a \big]\cdot\\
&\Pr(x \ne 0, T=a, \delta_1=0) \label{eq: JSQ-1}\\  
\nonumber 
\leq & \frac{L-3}{2^\lambda-1} \Pr(x \ne 0, T=a, \delta_1=0),
\end{align} 
}
where Eq. (\ref{eq: JSQ-1}) follows from the fact that $\delta_1$ is determined by $T$.   Hence, 
\begin{align}\nonumber
&\sum_{a}P_{xy}[(m',T')  \mbox{ valid} , T=a, x \ne 0, \delta_1=0]\\
&\le \frac{L-3}{2^\lambda-1} \Pr(x \ne 0, \delta_1=0).  \label{eq: JSQ-2} 
\end{align}

{\bf Case $\delta_1\ne 0$. }  In this case, $T'-T=\delta_1x^{L-1}+\delta_1^2x^{L-2}+Q(x)+\delta_1y,$ where $Q(x)$ is some polynomial of degree at most $L-3$ (when fix $T=a$).  Here we use the fact:  since  $GF(2^\lambda)$ has character is 2 and $L=3$ mod 4, it follows that both $L$ and $(L-1)L/2$ are  1 in $GF(2^\lambda).$ Representing $y$ in terms of $x$ and substituting it into $T$, we have $a=T=\delta_1x^{L-1}+\mu(x)$ for some polynomial $\mu(x)$ of degree at most $L-3$ (when fix $T=a$). There are at most $L-1$ possible $x$ to satisfy this. Further, when  $T=a$ and $x$ are fixed with $x\ne 0$, there is a unique $y$ satisfying $T=a$. Hence, given $T=a$ and the fact $x\ne 0$, we know that $x$ is  uniformly random over $GF(2^\lambda)-\{0\}.$ Finally, we again remind that $T', m', \delta_1, \delta_2$ are determined by $T$ (hence, given $T=a$, $(T', m', \delta_1, \delta_2)$ are all fixed). With these facts in mind,   we have (similar to Case $\delta_1=0$)
{\footnotesize
\begin{align} 
\nonumber 
&P_{xy}[(m',T')  \mbox{ valid} \mid  x \ne 0, T=a, \delta_1\ne 0]\cdot \Pr(x \ne 0, T=a, \delta_1\ne 0) \\\nonumber 
\leq & P_{xy}\big[\delta_1x^{L-1}+\delta_1^2x^{L-2}+Q(x)+\delta_1y=T'-a  \mid  x \ne 0, T=a \big]\cdot \\
&\Pr(x \ne 0, T=a, \delta_1=0) \label{eq: JSQ-3}\\  
\nonumber 
\leq & \frac{L-1}{2^\lambda-1} \Pr(x \ne 0, T=a, \delta_1\ne 0),
\end{align} 
}
 Hence, 
\begin{align}\nonumber
&\sum_{a}P_{xy}[(m',T')  \mbox{ valid} , T=a, x \ne 0, \delta_1=0]\\
&\le \frac{L-1}{2^\lambda-1} \Pr(x \ne 0, \delta_1\ne 0).  \label{eq: JSQ-4} 
\end{align}
Combining Eq. (\ref{eq: JSQ-2})(\ref{eq: JSQ-4}), we have 
$\sum_{a}P_{xy}[(m',T')  \mbox{ valid} , T=a, x \ne 0]\le \frac{L-1}{2^\lambda-1} \Pr(x \ne 0).$  Notice $\Pr(x\ne 0)=\frac{2^\lambda-1}{2^\lambda}. $ 
We know that $\sum_{a}P_{xy}[(m',T')  \mbox{ valid} , T=a, x \ne 0]\le \frac{L-1}{2^\lambda}.$ Plugging in Eq. (\ref{eq: JSQ-0}), we obtain  our result. 
\end{proof}

\subsection{\textbf{Proof of Theorem~\ref{THM:FUZZYEXTOTCONST1}}}\label{pf:THM:FUZZYEXTOTCONST1}

Let us start the security analysis with some preparation results. 
Claim 1 and Lemmas \ref{le: dist}-\ref{le: d} are based on quite standard techniques of probabilistic distances. We give proofs in Appendix \ref{appendix: C} for completeness.

The following claim is a well-known  fact.  

\vspace{0.05in} \noindent {\bf Claim 1. }   Let $X$ and $Y$ be random variables over ${\cal Z}$.
Assume that  $F: {\cal Z}\rightarrow {\cal V}$ is a deterministic  function. Then
$\Delta(F(X); F(Y))\le \Delta(X; Y).$

\vspace{.10in} This claim is correct if $F$ is randomized as it holds for each fixed randomness. The following lemma can be proven by triangle inequality and induction on $\mu.$

\begin{lemma} Let $W$ be a $(\alpha, m, N)$-source of length $n$ over $\{0, 1\}$. Let $E: \{0, 1\}^m \times \{0, 1\}^r\rightarrow \{0, 1\}^\nu$ be an average  $(m, \alpha, \nu,   \epsilon)$-randomness extractor. Let $Z$ be a uniform $r$-bit string and  $A_i$ be uniformly random subset of $[n]$ of size $m$ for $i=1, \cdots, \mu.$ Let $d_i=E(W[A_i], Z)$ for $i=1, \cdots, \mu$. Assume  $E(\tilde{U}, Z)=U$, where $\tilde{U}$ is uniformly random over $\{0, 1\}^m$ and $U$ is the uniformly distributed over $\{0, 1\}^{\nu}$. If $\mu m<N,$ then,
$\Delta({\bf d}, Z, {\bf A}; U^\mu, Z, {\bf A})\le \mu\epsilon, $ where ${\bf A}=(A_1, \cdots, A_\mu)$ and  ${\bf d}=(d_1, \cdots, d_\mu).$  \label{le: dist}
\end{lemma}

\begin{lemma} Let $W=W_1, \cdots, W_n$ be an $(\alpha,  m, N)$-source of length $n$ over alphabet $\{0, 1\}$. Let $E: \{0, 1\}^m \times \{0, 1\}^r\rightarrow \{0, 1\}^\nu$ be an average  $(m, \alpha, \nu,   \epsilon$)-randomness extractor.
Let $A_i$ be uniformly random subset of $[n]$ of size $m$ for $i=1, \cdots, \ell$ and $Z$ is a uniform $r$-bit string. 
Let $p_i=E(W[A_i], Z)\oplus 0^t|S$ for $S\leftarrow \{0, 1\}^{\nu-t}$, $i=1, \cdots, \ell$.
Then,
$
\Delta(S, {\bf p}, Z, {\bf A}; U, {\bf p}, Z, {\bf A})\le 2\ell \epsilon,$
where ${\bf p}=(p_1, p_2, \cdots, p_\ell), {\bf A}=(A_1, \cdots, A_\ell)$, $\ell m<N$  and $U\leftarrow \{0, 1\}^{\nu-t}$. \label{le: d}
\end{lemma}

We now prove the construction~\ref{const1} is a reusable fuzzy extractor satisfying definition~\ref{def:rfuzzyext}. 


To prove that the construction is a reusable FE, we need to prove the $\epsilon'$-correctness according  to Definition~\ref{def:fuzzyext} and then reusability satisfying Definition~\ref{def:rfuzzyext}. 

\textit{Correctness.}
    We first consider correctness error. That is, when there is no attack, the ciphertext will be accepted with high probability.  
Note that the maximum Hamming distance between $W$ and $W'$ is $t'$ i.e. $d(W,W') \le t'$. Observe that for any $i$, \\
$Pr[(W'_{j_{i,1}},W'_{j_{i,2}},\cdots,W'_{j_{i,m}})=(W_{j_{i,1}},W_{j_{i,2}},\cdots,W_{j_{i,m}})] \\ \ge (1 - t'/(n-m))^m$, where $A_i=\{i_1, \cdots, i_m\}=\{j_{i,1},j_{i,2},\cdots,j_{i,m}\}$ and $1 \leq j_{i,1},j_{i,2},\cdots,j_{i,m} \leq n$. This holds true since  $Pr[W'_{j_{i,1}}=W_{j_{i,1}}] \ge (1-t'/n)$ as ${j_{i,1}}$ is uniform. Next, given $W'_{j_{i,1}}=W_{j_{i,1}}$,  the probability that  $W'_{j_{i,2}}=W_{j_{i,2}}$ is at least $1-t'/(n-1)$. We thus have  \\$Pr[(W'_{j_{i,1}},W'_{j_{i,2}},\cdots,W'_{j_{i,m}})=(W_{j_{i,1}},W_{j_{i,2}},\cdots, W_{j_{i,m}})]\\ \ge (1-t'/n)(1-t'/(n-1))\cdots(1-t'/(n-m)) \\ \ge (1 - t'/(n-m))^m$. Consequently, the probability that none of the $v_i$'s matches at the receiver is at most $(1-(1-\frac{t'}{n-m})^m)^\ell$, where $v_i=(W'_{j_{i,1}},W'_{j_{i,2}},\cdots,W'_{j_{i,m}})$, and $i \in \{1,2,\cdots,\ell\}$.

In addition, $Rep$ in figure~\ref{fig:const1} may return an incorrect extracted key due to an error in step $(ii)$. This error might occur due to a collision that we now explain in detail. A collision at step $(ii)$ may occur with probability $(\frac{1}{2^{t}})$, assuming $E$ is implemented as a function from a 
universal hash family.  
Moreover, $Rep$ may return an incorrect extracted key if there is at least one collision at step $(ii)$.
This may happen with any of the $\ell$ executions of step $(ii)$ and hence with probability at most 
$\ell \cdot 2^{-t}$.

Since the fuzzy extractor's allowable error parameter is $\epsilon'$, we need to choose $\ell$ so that 

\begin{equation}
(1-(1-\frac{t'}{n-m})^m)^\ell + \ell \cdot 2^{-t} \leq \epsilon'.
\end{equation}

\textit{Reusability.} We first prove the case when $\eta=1$, and  then  extend it to the case when  $\eta > 1$.

{\bf Case $\eta=1$}. We need to prove that ${\rm \Delta}((R, P), (U_{\xi}; P)) \le  \sigma$, where $P$ is the public strings.

In step $(i)$ of  Gen() algorithm described in figure~\ref{fig:const1}, we sample a random subset $A_i=\{i_1, \cdots, i_m\}$ from $[n]$. In step $(ii)$, $E$ is implemented by $H$.

Since $W$ is a $(\alpha, m, N)$-sample source, we have $\tilde{H}_\infty(W[A] \mid W[B], A, B)\ge \alpha$, where $A$ (resp. $B$) is a purely random subset of $[n]$ of size $m$ (resp. $N$).


Since $H:\mathcal{X} \times \mathcal{S} \rightarrow \{0,1\}^\nu$ is a  
universal hash family, from lemma~\ref{glhl}, if $\nu \leq \alpha + 2 - 2\log(\frac{2\ell}{\sigma})$, $H$ is an $(m,\alpha,\nu,\frac{\sigma}{2\ell})$-extractor.  Now, $\nu = \xi + t$. Thus, if $\xi \leq \alpha + 2 - 2\log(\frac{2\ell}{\sigma}) - t$, $H$ is an $(m,\alpha,\nu,\frac{\sigma}{2\ell})$-extractor. 


Since $H$ is an $(m,\alpha,\nu,\frac{\sigma}{2\ell})$-extractor, from lemma~\ref{le: dist}, \ref{le: d}, if $\xi \leq \alpha + 2 - 2\log(\frac{2\ell}{\sigma}) - t$ and $\ell m<N$, considering $S=R$ in lemma~\ref{le: d},   we have
\begin{equation}
\label{sdadvconst11}
\Delta(R, {\bf p}, Z, {\bf A}; U, {\bf p}, Z, {\bf A})\le \sigma.
\end{equation}

Therefore, $(Gen,Rep)$ is a $(\mathcal{V}^n,\mathcal{W},\xi,t',\epsilon',\sigma)$-fuzzy extractor, where  $(1-(1-\frac{t'}{n-m})^m)^\ell + \ell \cdot 2^{-t} \leq \epsilon'$, $\ell m<N$ and $\mathcal{V} \in \{0,1\}$.
Furthermore, the length of the extracted key will be $\xi \leq \alpha + 2 - 2\log(\frac{2\ell}{\sigma}) - t$. 

{\bf Case $\eta>1$}. 
 In response to a query to $Gen$ oracle (i.e. generation oracle query), the oracle returns a pair of key $r^i$ and ciphertext $c^i$ to the adversary, where $r^i=R^i$ and $c^i=(p_1|A_1,\cdots, p_\ell|A_\ell)^i$ according to the $Gen(W^i)$\remove{MOD.$iK_{RFE}.Enc(W)$} procedure described in\remove{MOD.Algorithm~\ref{alg:fuzzygeneration1}} figure~\ref{fig:const1}. Now since $W$ is a source with $(\alpha,m,N)$-samples, in each query to generation\remove{encapsulation} oracle, $Gen(W^i)$\remove{MOD.$iK_{RFE}.Enc(W)$} procedure runs with new samples with conditional entropy $\alpha$. Hence, the uncertainty about the new samples remains the same before and after  $\eta$\remove{MOD.$q_e$} queries to the generation\remove{encapsulation} oracle from adversary's perspective. Therefore, the entropy of the new samples remains same. Now proceeding in similar manner as the proof for the {\bf case $\eta=1$},   
 we can prove that, if \\$\xi \leq \alpha + 2 - 2 \cdot \log(\frac{2 \ell}{\sigma}) - t$\remove{MOD.$\xi \leq \alpha + 2 - 2 \cdot \log(\frac{4\cdot \ell}{\sigma}) - t - 2\lambda$}, then the $(Gen,Rep)$\remove{MOD.iKEM $iK_{RFE}$} described in figure~\ref{fig:const1}\remove{MOD.Algorithm~\ref{alg:fuzzygeneration1} and \ref{alg:fuzzyreproduction1}} is $(\eta,\sigma)$-reusable $(\mathcal{V}^n,\mathcal{W},\xi,t',\epsilon',\sigma)$-fuzzy extractor, where \\  
$(1-(1-\frac{t'}{n-m})^m)^\ell + \ell \cdot 2^{-t} \leq \epsilon'$ and $\ell \eta m<N$. 
$\hfill\qed$

\subsection{\textbf{Proof of Theorem~\ref{THM:FUZZYEXTOTCONST2}}}\label{pf:THM:FUZZYEXTOTCONST2}
Lemma~\ref{le:5} is based on quite standard techniques of probabilistic distances. We give its proof in Appendix \ref{appendix: C} for completeness.

\begin{lemma} $\Delta(R, {\bf p}, Z, {\bf A}, T; U, {\bf p}, Z, {\bf A}, T)\le 4\ell \epsilon$. \label{le:5}
\end{lemma}

We now prove the construction~\ref{const2:cca} is a strongly robust and reusable fuzzy extractor satisfying definition~\ref{def:srrfuzzyext}.


We need to prove the $\epsilon'$-correctness according  to Definition~\ref{def:fuzzyext} and then reusability satisfying Definition~\ref{def:rfuzzyext} and finally prove robustness.

{\em Correctness. } We first consider correctness error. 
Proceeding in the same way as the proof of Theorem~\ref{THM:FUZZYEXTOTCONST1}, we can prove that the probability that no $v_i$ matches at the receiver is at most $(1-(1-\frac{t'}{n-m})^m)^\ell$, where $v_i=(W'_{j_{i,1}},W'_{j_{i,2}},\cdots,W'_{j_{i,m}})$, and $i \in \{1,2,\cdots,\ell\}$.

In addition, $Rep$ may be incorrect due to an error in step $(ii)$ and verification of $T'=T$. These errors might occur due to collision. A collision at step $(ii)$ might occur with probability $(\frac{1}{2^{t}})$, assuming $E$ is implemented as a function from a 
universal hash family. In addition, similar to lemma~\ref{le: forge} (except $\delta_1|\delta_2=(d_i\oplus d_i')_{\nu-2\lambda+1\cdots \nu}$ with $d_i$ from $W$ and $d'_i$ from $W'$), a collision $T'=T$ occurs with probability $L2^{-\lambda}$. Now, $Rep$ may be incorrect if there is at least one collision at step $(ii)$ together with a collision $T'=T$. This may happen with any of the $\ell$ iterations and hence with probability at most 
$\ell \cdot 2^{-t} \cdot L\cdot 2^{-\lambda}$. 

Since the fuzzy extractor's allowable error parameter is $\epsilon'$, we need to choose $\ell$  so that 
\begin{equation}
(1-(1-\frac{t'}{n-m})^m)^\ell + \ell \cdot 2^{-t} \cdot L\cdot 2^{-\lambda} \leq \epsilon'.    
\end{equation}

\textit{Reusability.} Let $\eta=q_e+1$. We first prove the case when $\eta=1$, and  then  extend it to the case when  $\eta > 1$.

{\bf Case $\eta=1$}.   We need to show that \\${\rm \Delta}((R, P), (U_{\xi}; P)) \le  \sigma$, where $P$ is the public strings. In step $(i)$ of  $Gen$ algorithm of Figure~\ref{fig:const12}
, a random subset $A_i=\{i_1, \cdots, i_m\}$ is sampled  from $[n]$, and in step $(ii)$, $E$ is implemented by $H$.
Note that $W$ is a $(\alpha, m, N)$-sample source. Hence, we have $\tilde{H}_\infty(W[A] \mid W[B], A, B)\ge \alpha$, where $A$ (resp. $B$) is a purely random subset of $[n]$ of size $m$ (resp. $N$).


Since $H:\mathcal{X} \times \mathcal{S} \rightarrow \{0,1\}^\nu$ is a  
universal hash family, from lemma~\ref{glhl}, if $\nu \leq \alpha + 2 - 2\log(\frac{4\ell}{\sigma})$, $H$ is an $(m,\alpha,\nu,\frac{\sigma}{4\ell})$-extractor.  Now, $\nu = 2\lambda + \xi + t$. Thus, if $\xi \leq \alpha + 2 - 2\log(\frac{4\ell}{\sigma}) - t - 2\lambda$, $H$ is an $(m,\alpha,\nu,\frac{\sigma}{2\ell})$-extractor. 


Since $H$ is an $(m,\alpha,\nu,\frac{\sigma}{4\ell}$-extractor, from lemma~\ref{le: dist}, \ref{le: d} and \ref{le:5}, if $\xi \leq \alpha + 2 - 2\log(\frac{4\ell}{\sigma}) - t - 2\lambda$ and $\ell m<N$, we have
\begin{equation}
\label{sdadvconst2}
\Delta(R, {\bf p}, Z, {\bf A}, T; U, {\bf p}, Z, {\bf A}, T)\le \sigma.
\end{equation}

Therefore, $(Gen,Rep)$ is a $(\mathcal{V}^n,\mathcal{W},\xi,t',\epsilon',\sigma)$-fuzzy extractor, where  \\$(1-(1-\frac{t'}{n-m})^m)^\ell + \ell \cdot 2^{-t} \cdot L \cdot 2^{-\lambda} \leq \epsilon'$, $\ell m<N$ and $\mathcal{V} \in \{0,1\}$.
Furthermore, the length of the extracted key will be $\xi \leq \alpha + 2 - 2\log(\frac{4\ell}{\sigma}) - t - 2\lambda$.

{\bf Case $\eta  > 1$}.
The proof follows similar arguments of {\it reusability} part of the proof of Theorem~\ref{THM:FUZZYEXTOTCONST1} and the {\it security} arguments described above.
In response to a query to $Gen$ oracle, the oracle returns a pair of key $r^i$ and ciphertext $c^i$ to the adversary, where $r^i=R^i$ and $c=(p_1,\cdots, p_\ell, T)$ according to the $Gen(W^i)$\remove{MOD.$iK_{RFE}.Enc(W)$} procedure described in
 Figure~\ref{fig:const12}. As $W$ is a source with $(\alpha,m,N)$-samples, in each $Gen$ oracle query,  $Gen(W^i)$\remove{$iK_{RFE}.Enc(W)$} procedure runs with new samples with conditional entropy $\alpha$.\remove{MOD., where $\tilde{H}_\infty(W[A] \mid W[B], A, B)\ge \alpha$ and $A$ (resp. $B$) a purely random subset of $[n]$ of size $m$ (resp. $N$).} Hence, the uncertainty about the new samples remains the same before and after $\eta$ $Gen$ oracle (or generation oracle) queries from the adversary's perspective. Therefore, the entropy of each new samples remains same. Now proceeding in similar manner as the proof for the {\bf  `case $\eta=1$'}, 
 we can prove that, if $\xi \leq \alpha + 2 - 2 \cdot \log(\frac{4\cdot \ell}{\sigma}) - t - 2\lambda$, then the $(Gen,Rep)$\remove{MOD.iKEM $iK_{RFE}$} described in Figure~\ref{fig:const12}
 is $(\eta,\sigma)$-reusable $(\mathcal{V}^n,\mathcal{W},\xi,t',\epsilon',\sigma)$-fuzzy extractor, where \\
$(1-(1-\frac{t'}{n-m})^m)^\ell + \ell \cdot 2^{-t} \cdot L\cdot 2^{-\lambda} \leq \epsilon'$ and $\ell \eta m<N$. 
\remove{$\epsilon'$-correct and $\sigma$-IND-$q_e$-CEA secure, where  $(1-(1-\frac{t'}{n-m})^m)^\ell + \ell \cdot 2^{-t} \cdot \frac{L+1}{2^\lambda} \leq \epsilon'$.} 

{\em Robustness. } Robustness proof is given in Appendix \ref{pf:intcxtx} 
$\qed$

\subsection{\textbf{Robustness Proof for  Theorem  \ref{THM:FUZZYEXTOTCONST2}}}\label{pf:intcxtx}
We need to prove our scheme satisfies Definition \ref{def:srrfuzzyext}.  We show that through $q_e$ generation  queries and $q_d$ reproduction  queries, the robustness  is broken negligibly. Assume ${\cal A}$ is the robustness  attacker. Upon generation   and reproduction  queries, challenger acts normally. Denote this game $\Gamma_0.$  We now revise $\Gamma_0$ to $\Gamma_1$ so that $d_i=E(W[A_i], Z)$ is replaced by $d_i\leftarrow \{0, 1\}^{\nu}$ in generation  query or reproduction  query (for some $A_i$ that does not appear in a previous query; otherwise, use the existing $d_i$).    Let {\bf Succ}$(\Gamma)$ be the success event of ${\cal A}$ in game $\Gamma.$
Let $d_i$ be taken from $d_i=E(W[A_i], Z)$ or $d_i\leftarrow \{0, 1\}^{\nu}$, for all $i.$ Then, we consider a distinguisher that distinguishes $\Gamma_0$ and $\Gamma_1$. Upon receiving $d_i$ that is either $d_i=E(W[A_i], Z)$ or uniformly random, challenger prepares the public sampling randomness $r$ that results in sampling $A_i$'s for all $i$'s in the challenge. Now challenger simulates $\Gamma_0$ with ${\cal A}$ against it by providing ${\cal A}$ with $Z$ and public randomness $r$, except all $d_i$'s are from his challenge tuple. Finally, if ${\cal A}$ succeeds, output $0$; otherwise, output 1. The simulated game is a randomized function with input ${\bf d}_1, \cdots, {\bf d}_{q_d+q_e}$ and a binary output 0 or 1,  where ${\bf d}_i$ is the vector of $(d_1, \cdots, d_\ell)$ in the generation  or a reproduction  query (with new $A_i$'s).   Denote this function by $G({\bf d}).$ Since $A_1, \cdots, A_\ell$ by challenger or adversary is from public random string, $W[A_i]$ follows the distribution of $(\alpha, m, N)$-source $W$.  By claim 1 and Lemma \ref{le: dist}, we immediately have the following.

\begin{lemma} $|P({\bf Succ}(\Gamma_0))-P({\bf Succ}(\Gamma_1))|\le (q_d+q_e)\ell\epsilon.$ \label{le: gap01}
\end{lemma}

Now we consider the success event {\bf Succ}$(\Gamma_1)$. We assume ${\cal A}$ will not query the output of generation  query to the reproduction  oracle as ${\cal A}$ already knows the answer. Consider the $i$th reproduction  query $\{A_i|p_i\}_{i=1}^\ell|T$. Let bit $E_i$ be the decision bit for the reproduction  query (0 for reject and 1 for success). We modify $\Gamma_1$ to $\Gamma_2$ such that if upon the first $E_i=1$, then stop the security game. Obviously, $P({\bf Succ}(\Gamma_1))=P({\bf Succ}(\Gamma_2))$. Let $E_i^*$ be the event $E_i=1$ while $E_j=0$ for $j<i.$ So $P({\bf Succ}(\Gamma_2))\le \sum_{i=1}^{q_d}P(E_i^*).$  We first bound $P(E_1^*).$

\vspace{0.10in} Let  this first reproduction  query be   $C_1=(A'_1, p'_1, \cdots, A'_\ell, p'_\ell, T')$, where $A_1',\cdots, A'_\ell$ are from public randomness $r$ (note:  they could be previously sampled in  the generation query processing).
Hence, $d_1, \cdots, d_\ell$ for this query are uniformly random (as we consider  $\Gamma_2$).  For simplicity, some generation  query has generated ciphertext  $C_1'$ using the same sample set $A'_1, \cdots, A'_\ell$ (otherwise, the proof will be similar and simpler). Assume  $d_1, \cdots, d_\ell$ generated  $p_1, \cdots, p_\ell, T$ for that generation  query.  We bound the probability that $T'$ generated by ${\cal A}$ is valid. Notice that each generation  query uses independent $d_i$'s and hence can be simulated by ${\cal A}$. So we can assume that ${\cal A}$ only issue one generation  query (which uses $A_1', \cdots, A_\ell'$). Denote this by $\Gamma_2^1$. Let the hash output be $T$ for the generation  query. Since the reproduction  query can not use the same ciphertext, assume $p_i'\ne p_i$. Hence, let $x|y=R_1$ in $p_i$ and let $\delta_1|\delta_2$ be the last $2\lambda$ bits in $p_i\oplus p_i'\oplus E(o[A_i], A),$ where $o=W'-W$ is assumed to be known to ${\cal A}$ (this will only increase the success probability of ${\cal A}$).    Then, the decrypted tag key from $C_1$ using $W'$ at verifier will be \\$(x+\delta_1)|(y+\delta_2)$. Hence, if $T'$ is valid, we know that $T=x^L+x^2m(x)+xy$ and $T'=(x+\delta_1)^L+(x+\delta_1)^2m'(x+\delta_1)+(x+\delta_1)(y+\delta_2)$, with $L=\lceil\nu \ell/\lambda\rceil+4$. Since now $x, y$ are uniformly random $\lambda$ bits and are independent of $(A_1|p_1|\cdots|p_\ell|A_\ell)$ due to the one-time pads $d_1, \cdots, d_\ell$) (hence idependent of the encoded vector ${\bf m}$),   by Lemma \ref{le: forge}, conditional on $T,$ the probability that  $T'$ is valid is  at most $2^{-\lambda}(L+1)$. Further, there are at most $\ell$ possible $i$. Thus, $P(E_1^*)\le 2^{-\lambda}\ell(L+1).$

Now we consider $P(E_k^*).$ To evaluate this, we consider a variant $\Gamma_2^3$ of $\Gamma_2$ where the $t$th reprodution  query for $t<k$ is decided as reject (without verifying the tag $T$ in its query). Let $\Sigma$ be the randomness of $\Gamma_2$. Then, if $\Sigma$ leads to reject for all previous $k-1$ reproduction  query, then adversary view in $\Gamma_2^3$ and $\Gamma_2$ is identical; otherwise, some {\bf reject} of some $t$th reproduction  query is wrong. But in this case,  $E_k^*$  will not occur. It follows $P(E_k^*(\Gamma_2))=P(E_k^*(\Gamma_2^3)).$ On the other hand, since the previous $k-1$ reproduction  query always results in reject, it can be simulated by ${\cal A}$ himself. Hence, $P(E_k^*(\Gamma_2^3))=P(E_1^*).$
It follows that $P({\bf Succ}(\Gamma_2))\le q_dP(E_1^*)\le q_d2^{-\lambda}\ell (L+1).$  From Lemma \ref{le: gap01} and $P({\bf Succ}(\Gamma_1))=P({\bf Succ}(\Gamma_2))$, it follows that $P({\bf Succ}(\Gamma_0))\le (q_d+q_e)\ell\epsilon+q_d2^{-\lambda}\ell (L+1).$ $\hfill\square$

\subsection{\textbf{Proofs of Basic Results in Section~\ref{construction:const1}} } \label{appendix: C} 
\noindent{\bf Proof of Claim 1. }  By calculation, we have

$
\begin{array}{ll}
&{\Delta}(F(X); F(Y))\\
=& \frac{1}{2} \sum_{v\in {\cal V}} |\Pr[F(X)=v]-\Pr[F(Y)=v]|\\
=& \frac{1}{2} \sum_{v\in {\cal V}} |\sum_{u: F(u)=v} (\Pr[X=u]-\Pr[Y=u])|\\
\le & \frac{1}{2} \sum_{v\in {\cal V}}\sum_{u: F(u)=v}|\Pr[X=u]-\Pr[Y=u]|\\
=& {\Delta}(X; Y). \hspace{2.30in}\hfill\square  
\end{array}
$ 

\vspace{.10in} \noindent{\bf Proof of Lemma} \ref{le: dist}.  Use induction. When $\mu=1$, it holds from assumption and Claim 1. If it holds for $\mu=k-1$, consider case $\mu=k.$ By assumption, $E(\tilde{U}, Z)=U.$ Note that $d_i=E(W[A_i], Z).$ We have
{\footnotesize
\begin{align*}
&\Delta({\bf d}, Z, {\bf A}; U^k, Z, {\bf A})\\
\le& \Delta(d_k, {\bf d}^{k-1}, Z, {\bf A}; U, {\bf d}^{k-1}, Z, {\bf A})+ \\
&\Delta(U, {\bf d}^{k-1}, Z, {\bf A}; U, U^{k-1}, Z, {\bf A})\\
\le & \Delta(d_k, \{W[A_i]\}_1^{k-1}, Z, {\bf A}; U, \{W[A_i]\}_1^{k-1}, Z, {\bf A})+\\
&\Delta({\bf d}^{k-1}, Z, {\bf A}; U^{k-1}, Z, {\bf A})\\
\le& \epsilon+(k-1)\epsilon=k\epsilon,
\end{align*}
}
  where the 2nd inequality follows from Claim 1; the first part of  the last inequality  follows from the definition of $E$ and $\tilde{H}_\infty(W[A_k]\mid \{W[A_i]\}_1^{k-1}, {\bf A})\ge \alpha$ (from definition of $W$); the second part follows from the induction assumption with the fact that $A_k$ is independent of the remaining random variable.  $\hfill\square$
  

\vspace{.10in} \noindent{\bf Proof of Lemma} \ref{le: d}.  Let $d_i=E(W[A_i], Z)$. Denote $d_i=X_i|Y_i$ with $X_i$ the first $t$ bits and $Y_i$ the remaining bits of $d_i$.  Let $U_1, U_2$ be uniformly random variables in the domain of $X_1, Y_1$ respectively. Let ${\bf U}_1, {\bf U}_2$ be $\ell$ iid copies of $U_1, U_2$ respectively. For vector ${\bf v}$ and variable $\alpha$,
${\bf v}\oplus \alpha$ denotes $(v_1\oplus \alpha, \cdots, v_\ell\oplus \alpha).$ For simplicity, let $C=(Z, {\bf A})$.
\begin{align*}
&\Delta({S, {\bf p}, C}; U, {\bf p}, C)\\
&=\Delta(S, {\bf d}, C; U, {\bf d}', C), \mbox{ where } {\bf d}'={\bf d}\oplus 0^t|S\oplus 0^t|U\\
&=\Delta(P_{S {\bf d} C}; P_UP_{{\bf d}' C})\\
&= \Delta(P_{{\bf d} C}; P_{{\bf d}' C}),  \mbox{ (as $S$ is independent of $({\bf d}, C)$  and } \\
&\quad \mbox{ distributed as $U$)}  \\
&= \Delta(P_{{\bf XY}C}; P_{{\bf X}, {\bf Y}\oplus V, C}), \mbox{ (let $V\stackrel{def}{=}U\oplus S$))}\\
&\le \Delta(P_{{\bf XY}C}; P_{{\bf U}_1 {\bf U}_2 C})+\Delta(P_{{\bf U_1U_2}C}; P_{{\bf X}, {\bf Y}\oplus V, C})\\
&\le \ell \epsilon+\Delta(P_{{\bf U_1U_2}C}; P_{{\bf X}, {\bf Y}\oplus V, C}), \mbox{ (Lemma \ref{le: dist})} \\
&\le \ell \epsilon+ \Delta(P_{{\bf U_1U_2}CV}; P_{{\bf X}, {\bf Y}\oplus V, C, V})\\
&= \ell \epsilon+ \Delta(P_{{\bf U_1,U_2}-V, C, V}; P_{{\bf X}, {\bf Y}, C, V})\\
&= \ell \epsilon+ \Delta(P_{{\bf U_1U_2} C V}; P_{{\bf X} {\bf Y} C V})\\
&\quad \mbox{ (${\bf U}_2$ and ${\bf U}_2-V$ are identical, independent of the  }  \\
&\quad \mbox{ remaining variables)}\\
&= \ell \epsilon+ \Delta(P_{{\bf U_1 U_2} C}; P_{{\bf X} {\bf Y} C}), \mbox{ ($V$ is independent of the } \\
&\quad \mbox{ remaining variables)}\\
&\le 2\ell \epsilon, \mbox{ (Lemma \ref{le: dist})   $\hfill\square$} 
\end{align*}

\vspace{.10in} \noindent{\bf Proof of Lemma} \ref{le:5}. \quad  In Lemma \ref{le: d}, $S=R|R_1$ in our scheme. Thus,
$$\Delta(R|R_1, {\bf p}, Z, {\bf A}; U|U_1, {\bf p}, Z, {\bf A})\le 2\ell \epsilon,$$ where $U$ (resp. $U_1$) is  uniform $\xi$-bit (resp.  $2\lambda$-bit). Let $C=({\bf p}, Z, {\bf A}).$ By  claim 1, let $F(R|R_1, C)=(R, C,  T).$ Then,
$\Delta(R, C, T; U, C, U')\le 2\ell\epsilon,$
where $U'=tag(U_1, C)$ and $tag(\cdot)$ is the tag algorithm in our scheme used to compute  $T$.
Notice that
{\small
\begin{align*}
&\Delta(R, C, T; U, C, T)\\
\le & \Delta(R, C, T; U, C, U')+\Delta(U, C, U'; U, C, T)\\
\le & \Delta(R|R_1, C; U|U_1, C)+\Delta(U, C, U'; U, C, T), \mbox{ (Claim 1)}\\
\le & 2\ell \epsilon+\Delta(C, U'; C, T), \mbox{ (Lemma \ref{le: d};  $U$ is ind of ${\bf p}$ and $T, U'$)} \\
\le & 2\ell\epsilon+\Delta(U_1, C; R_1, C),\quad  \mbox{ (Claim 1)} \\
\le &4\ell \epsilon,\quad  \mbox{ (Claim 1 and Lemma \ref{le: d})}
\end{align*}
}
This completes our proof.  $\hfill\square$

\subsection{Figure for construction~\ref{const2:cca}}\label{appendix_fig_ccaITconst}

Figure~\ref{fig:fuzzyext} depicts pictorial representation of construction~\ref{const2:cca}.
\begin{figure}[!ht]%
\footnotesize
    \centering
\subfigure[\label{fig:gen1}]{\includegraphics[width=5.6cm]{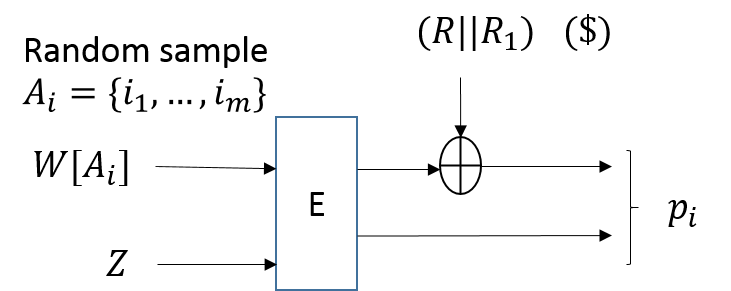}}%
    \quad 
    \subfigure[\label{fig:gen2}]{{\includegraphics[width=5.6cm]{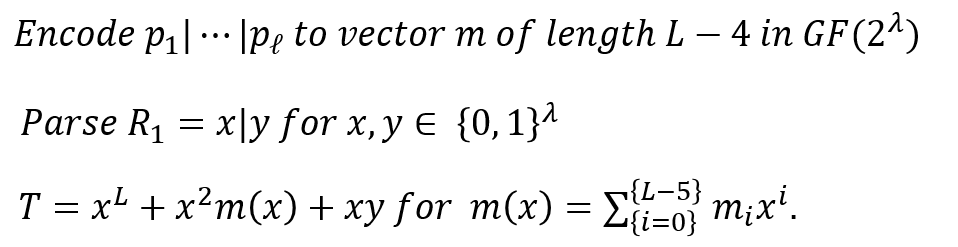} }}%
    \quad 
    \subfigure[\label{fig:rep1}]{{\includegraphics[width=5.6cm]{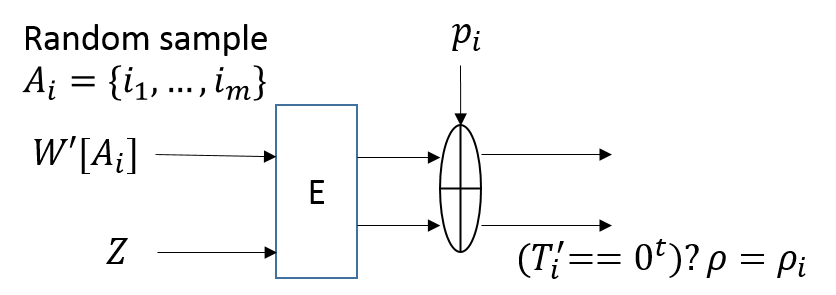} }}%
    \quad 
    \subfigure[\label{fig:rep2}]{{\includegraphics[width=5.6cm]{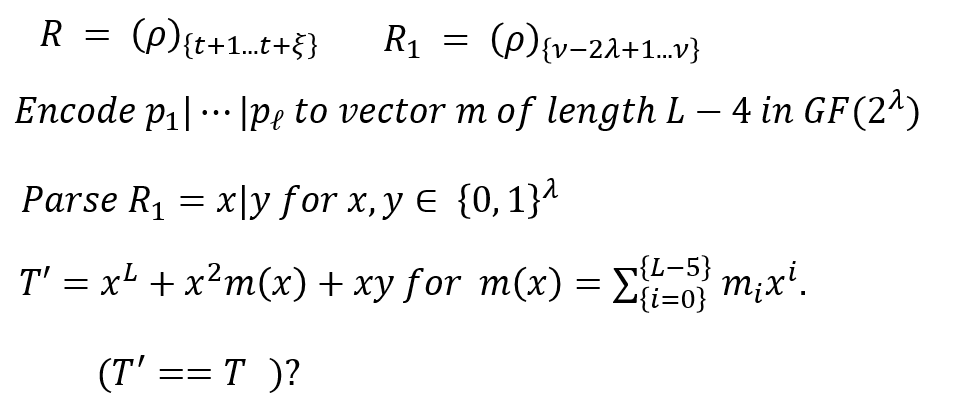} }}%
    \caption{\footnotesize High level diagram of fuzzy extractor construction~\ref{const2:cca}. Both $Gen(W)$ and $Rep(W')$ procedure are split into two parts. \ref{fig:gen1}: The first part of $Gen(W)$. This part iterates $\ell$ times. \ref{fig:gen2}: The second part of $Gen(W)$. This step is executed after completion of $\ell$ iterations of \ref{fig:gen1}. \ref{fig:rep1}: The First part of $Rep(W')$. This step is executed at most $\ell$ times until one match $(T'_i==0^t)$ is found. \ref{fig:rep2}: The second part of $Rep(W')$. This part is executed after one match in \ref{fig:rep1}. If verification of $(T'==T)$ fails, the algorithm continues from the part \ref{fig:rep1} again. The upper line of each extractor  represents higher order $\xi+2\lambda$  bits of its output i.e. $(E(.))_{t+1...\nu}$. The lower line of each extractor represents lower order $t$ bits of its output i.e. $(E(.))_{1...t}$. Extracted key is $R$. Randomness $Z$ is public}%
    \label{fig:fuzzyext}%
\end{figure}

\end{document}